\documentclass[11pt]{article}

\usepackage[margin=1in]{geometry}
\usepackage{amssymb,amsmath,amsthm,amsfonts,thmtools,mathtools}
\usepackage{relsize}
\usepackage{color}
\usepackage[export]{adjustbox}
\usepackage{yfonts}
\usepackage{mathrsfs}
\usepackage{authblk}

\usepackage{xspace}

\usepackage{verbatim}

\usepackage{graphicx,lipsum}
\usepackage[font=small]{caption}
\usepackage{subcaption}
\usepackage[utf8]{inputenc}

\usepackage{url}
\usepackage[all]{xy}
\usepackage{rotating}
\usepackage{ifpdf}
\usepackage{tcolorbox}
\usepackage{lipsum}

\usepackage{mdframed}             
\usepackage{xifthen} 
\usepackage{mathtools}  

\DeclareMathAlphabet\mathbfcal{OMS}{cmsy}{b}{n}

\usepackage{comment}
\usepackage[T1]{fontenc}

\usepackage{array}
\usepackage[shortlabels]{enumitem}

\usepackage{algorithm}
\usepackage[noend]{algpseudocode}

\usepackage{wrapfig} 

\input{insbox}

\newcommand{\ignore}[1]{}

\newcommand{\margincomment}[2]%
{\marginpar{\footnotesize\raggedright {\color{red}#1}: #2}}

\newcommand{\etal}{et~al.}

\newcommand{\myparagraph}[1]{{\medskip\noindent\textbf{#1}}}

\newcommand{\myitparagraph}[1]{{\medskip\noindent\textit{#1}}}

\newcommand{\myif}{\textbf{if}\xspace}
\newcommand{\mythen}{\textbf{then}\xspace}
\newcommand{\myelse}{\textbf{else}\xspace}

\newcommand{\mydo}{\textbf{do}\xspace}
\newcommand{\myforeach}{\textbf{for each}\xspace}

\newcommand{\ListLengths}{\setlength{\itemsep}{0ex}\setlength{\topsep}{1ex}\setlength{\partopsep}{0ex}}

\newtheorem{theorem}{Theorem}
\newtheorem{lemma}{Lemma}[section]
\newtheorem{corollary}[lemma]{Corollary}

\newtheorem{claim}[lemma]{Claim}
\newtheorem{observation}[lemma]{Observation}

\newenvironment{claimproof}[1]{\par\noindent\underline{Proof:}\space#1}{\hfill $\lhd$}

\newcommand{\braced}[1]{{ \left\{ {#1} \right\} }}

\newcommand{\suchthat}{{\;:\;}}

\newcommand{\VertexCover}{\textsc{VertexCover}}
\newcommand{\Partition}{\textsc{Partition}}

\newcommand{\Multicut}{\textsc{MCut}}
\newcommand{\MulticutSum}{\Multicut_{\Sigma}} 
\newcommand{\MulticutMax}{\Multicut_{\max}} 
\newcommand{\MultiwayCut}{\textsc{MwCut}} 
\newcommand{\MultiwayCutSum}{\MultiwayCut_{\Sigma}} 
\newcommand{\MultiwayCutMax}{\MultiwayCut_{\max}}

\newcommand{\ConMultiwayCut}{\textsc{Con}\MultiwayCut}
\newcommand{\ConMultiwayCutMax}{\ConMultiwayCut_{\max}} 

\newcommand{\MulticutTrees}{\textsc{Tree}\Multicut}
\newcommand{\MultiwayCutTrees}{\textsc{Tree}\MultiwayCut} 
\newcommand{\MulticutSumTrees}{\MulticutTrees_{\Sigma}} 
\newcommand{\MultiwayCutSumTrees}{\MultiwayCutTrees_{\Sigma}}

\newcommand{\MultiwayCutMaxTrees}{\MultiwayCutTrees_{\max}} 

\newcommand{\ConMulticutTrees}{\textsc{ConTree}\Multicut}
\newcommand{\ConMulticutMaxTrees}{\ConMulticutTrees_{\max}}
\newcommand{\ConMultiwayCutTrees}{\textsc{ConTree}\MultiwayCut}
\newcommand{\ConMultiwayCutMaxTrees}{\ConMultiwayCutTrees_{\max}}
\newcommand{\ConMultiwayCutSumTrees}{\ConMultiwayCutTrees_{\Sigma}}

\newcommand{\LocalCutTrees}{\textsc{LCutTrees}}
\newcommand{\LocalCutTreesImprop}{\textsc{LCutTrees}^\circ}

\newcommand{\TreeCutting}{\textsc{TreeCutting}}

\newcommand{\problemP}{\mathbb{P}}

\algnewcommand\algorithmiclet{\textbf{let}}
\algnewcommand\Let{\State \algorithmiclet\ }
\algnewcommand\algorithmiccase{\textbf{case}}
\algdef{SE}[CASE]{Case}{EndCase}[1]{\algorithmiccase\ #1}{\algorithmicend\ \algorithmiccase}%
\algtext*{EndCase}%
\algnewcommand{\IfThenElse}[3]{
  \State
  \algorithmicif\ #1\ \algorithmicthen\ #2\ \algorithmicelse\ #3}

\algnewcommand{\LeftComment}[1]{\Statex \hspace{-1em} \(\triangleright\) \emph{#1}}

\newcommand{\dptable}{D}

\newcommand{\multicutM}{\calM}

\newcommand{\multicutQ}{\calQ}
\newcommand{\maxcost}{\mathrm{cost}_{\max}}
\newcommand{\sumcost}{\mathrm{cost}_{\Sigma}}
\newcommand{\maxopt}{\mathrm{opt}_{\max}}
\newcommand{\sumopt}{\mathrm{opt}_{\Sigma}}

\newcommand{\locopt}{\mathrm{opt}_{\mathrm{loc}}}
\newcommand{\locoptim}{\mathrm{opt}_{\mathrm{loc}}^\circ}
\newcommand{\vcover}{\mathrm{vc}}
\newcommand{\maxweight}{\weight_{\max}}
\newcommand{\totalweight}{W}

\newcommand{\treeplus}{T^+}
\newcommand{\jumps}{\calB}
\newcommand{\states}{\calS}
\newcommand{\subsets}{\mathscr{P}}
\newcommand{\add}{\mathbin{\mathsmaller{\sqcup}}}

\newcommand{\nullsol}{\textsc{null}}

\newcommand{\nd}{\mathop{\mathrm{nd}}}
\newcommand{\spcomp}{\mathrm{sp}}

\newcommand{\TwinClasses}{\mathfrak{N}}
\newcommand{\TypeGraph}{\mathfrak{C}}
\newcommand{\TwinEdges}{\mathfrak{F}}

\newcommand{\typeof}{\tau}
\newcommand{\itwin}[1]{#1^{\textrm{i}}}
\newcommand{\ctwin}[1]{#1^{\textrm{c}}}

\newcommand{\CTwinClasses}{\TwinClasses^\textrm{(c)}}
\newcommand{\ITwinClasses}{\TwinClasses^\textrm{(i)}}
\newcommand{\splits}{\mathfrak{X}}

\newcommand{\indsets}{\mathcal{I}}

\newcommand{\termyes}{\bullet}
\newcommand{\termno}{\circ}
\newcommand{\lambdaplus}{\lambda^\termyes}
\newcommand{\lambdaminus}{\lambda^\termno}

\newcommand{\Dplus}{D^\termyes}
\newcommand{\Dminus}{D^\termno}
\newcommand{\Xplus}{X^\termyes}
\newcommand{\Xminus}{X^\termno}

\newcommand{\Bplus}{B^\termyes}
\newcommand{\Bminus}{B^\termno}

\newcommand{\cutof}{\partial}

\newcommand{\Terminals}{{\Gamma}}
\newcommand{\weight}{{\omega}}
\newcommand{\boundary}{\partial}
\newcommand{\ConflictGraph}{\calC}
\newcommand{\ConflictEdges}{\calF}

\newcommand{\lca}{\textrm{LCA}}
\newcommand{\textfalse}{\textrm{false}}
\newcommand{\texttrue}{\textrm{true}}

\newcommand{\children}{\textit{ch}}
\newcommand{\child}{\textit{child}}

\newcommand{\naturals}{{\mathbb{N}}}

\newcommand{\reals}{{\mathbb{R}}}
\newcommand{\posreals}{{\mathbb R}^+}

\newcommand{\razy}{\,} 
\newcommand{\posvalue}[1]{\langle #1\rangle^+}

\newcommand{\NP}{\ensuremath{\mathbb{NP}}\xspace}

\newcommand{\half}{\textstyle{\frac{1}{2}}}

\newcommand{\barI}{\bar{I}}

\newcommand{\calB}{\ensuremath{\mathcal B}\xspace}
\newcommand{\calC}{\ensuremath{\mathcal C}\xspace}

\newcommand{\calF}{\ensuremath{\mathcal F}\xspace}

\newcommand{\calM}{\ensuremath{\mathcal M}\xspace}

\newcommand{\calQ}{\ensuremath{\mathcal Q}\xspace}

\newcommand{\calS}{\ensuremath{\mathcal S}\xspace}

\newcommand{\tildeO}{{\tilde{O}}}

\newcommand{\staru}{u^\ast}
\newcommand{\starv}{v^\ast}
\newcommand{\starrho}{\rho^\ast}

\graphicspath{{./}}

\title{Algorithms for Connected Multicuts of Weighted Trees\thanks{%
Research partially supported by NSF grant CCF-2153723 (Marek Chrobak),
grant 24-10306S of GA ČR (Petr Kolman),
and by AGATE project funded from the Horizon Europe
Programme under Grant Agreement No. 101183743 (Hans Raj Tiwari).
}}

\author[$\dagger$]{Sunny Atalig}
\author[$\dagger$]{Marek Chrobak}
\author[$\S$]{Petr Kolman}
\author[$\ast$]{Claire Mathieu}
\author[$\S$]{Hans Raj Tiwary}

\affil[$\dagger$]{University of California at Riverside, USA}
\affil[$\S$]{Charles University, Prague, Czech Republic}
\affil[$\ast$]{CNRS and Universit{\'e} Paris Cit{\'e}, Paris, France}

\begin{document}
	
\maketitle

\begin{abstract}
In the \emph{connected multicut problem for weighted trees}, the instance is an edge-weighted tree $T = (V,E,\weight)$,
a set of terminals $\Terminals\subseteq V$, and a set of pairwise ``conflicts''
between these terminals represented by edges of a conflict graph $\ConflictGraph = (\Terminals,\ConflictEdges)$.
A \emph{connected multicut} for this instance is a partition $\multicutM$ of the vertex set $V$ into subsets
called the \emph{components of $\multicutM$}, where each
component induces a subtree of $T$, contains at least one terminal, but does not contain any conflicting terminal pair.
Two cost functions for multicuts are considered: in the min-sum variant, we minimize the total weight of the edges between
different components of $\multicutM$, and in the min-max
variant, we minimize the maximum, over all components, of the total weight of edges leaving the component.

In the first part of the paper, we focus on algorithms for the min-max variant,
for which we give a $2$-approximation algorithm for the unweighted case,
a $(2+\epsilon)$-approximation algorithm for the general case, and a mildly exponential exact algorithm for the unweighted case.

In the second part of the paper, we study algorithms that take advantage of the structure of the conflict graph.
We first give a polynomial-time exact algorithm for the min-max variant when $\ConflictGraph$ is a clique
(the variant referred to in the literature as the \emph{multiway cut}).
Next, we generalize this result by providing polynomial-time exact algorithms, for both variants,
for instances whose conflict graph has constant \emph{split complexity} --- a new
parameter that refines the concept of neighborhood diversity previously studied in the context
of parameterized complexity. This unifies and generalizes
previous results for the min-sum variant when conflict graphs are
complete multipartite graphs or unions of cliques.
\end{abstract}

\newpage


\section{Introduction}
\label{sec: introduction}


In the \emph{connected multicut problem for weighted trees}, the instance is an edge-weighted tree $T = (V,E,\weight)$ with $n$ vertices,
a set of \emph{terminals} $\Terminals\subseteq V$, and a set of pairwise ``conflicts''
between these terminals represented by edges of a conflict graph $\ConflictGraph = (\Terminals,\ConflictEdges)$.
A \emph{connected multicut} for this instance is a partition $\multicutM$ of the vertex set into subsets
called the \emph{components of $\multicutM$}, where each
component induces a subtree of $T$, contains at least one terminal, but does not contain any conflicting terminal pair.
When $\ConflictGraph$ is a clique, that is when all terminals conflict with each other, 
multicuts are referred to as \emph{multiway cuts} in the literature.


\myitparagraph{Min-sum cost.}
One natural way to measure the cost of a multicut $\multicutM$ is by
the total weight of the edges between different components of $\multicutM$, often referred to as the \emph{min-sum} cost.
The problem of minimizing the min-sum cost of a connected multicut in trees, 
denoted here $\MulticutSumTrees$, has been well studied in the literature.
It is easy to see that the connectivity requirement is redundant for $\MulticutSumTrees$ --- any multicut without conflicts
in its components can be converted, without increasing cost, into a connected one.
Garg et al.~\cite{garg_etal_primal-dual_multicut_trees_1997} proved that $\MulticutSumTrees$ is $\NP$-hard. 
They also provided a primal-dual algorithm with approximation ratio of $2$, that (under standard assumptions)
is optimal, because their $\NP$-hardness proof shows that for star
graphs, even without weights, the problem is polynomial-time equivalent to the $\VertexCover$ problem.
As shown by Calinescu~\etal~\cite{calinescu_etal_multicuts_bounded_degree_treewidth_2003}, the unweighted version of $\MulticutSumTrees$ remains $\NP$-hard
even for binary trees, although for such trees an approximation ratio $1+\epsilon$ can be achieved, for each $\epsilon > 0$.
Yannakakis~\etal~\cite{yannakakis_etal_cutting_and_partitioning_1983} observed (prior to~\cite{garg_etal_primal-dual_multicut_trees_1997}) that $\MulticutSumTrees$
is polynomial-time solvable for a constant number of terminals.
The complexity of $\MulticutSumTrees$ parametrized by the number of edges in the optimum multicut has been extensively studied
\cite{guo_niedermeier_multicut_in_trees_2005,liu_etal_generalized_multiway_cut_in_trees_2012,bousquet_etal_kernel_multicut_in_trees_2009,
chen_etal_multicut_in_trees_vertex_cover_2011,kanj_etal_improved_algorithms_cut_problems_on_trees_2015}.

Some earlier work on $\MulticutSumTrees$ studied its subproblems with some restrictions on the structure of the conflict graph. 
The case when $\ConflictGraph$ is a clique is the multiway cut problem for trees,
for which, as observed earlier~\cite{chopra_rao_multiway_cut_polyhedron_1991,dahlhaus_1994_complexity_multiterminal_cuts,hasan_motwani_query_optimization_1995,costa_billionnet_multiway_cut_trees_2004},
the dynamic-programming approach yields an $O(n)$-time algorithm.
Motivated by applications to query optimization in database applications, 
Hasan and Motwani~\cite{hasan_motwani_query_optimization_1995} studied a more general model
when  $\ConflictGraph$ is a complete multi-partite graph, for which they gave an $O(n^2)$-time algorithm. 
Another natural extension of the clique case is when  $\ConflictGraph$ is a union of a constant number of cliques,
for which polynomial-time algorithms were given by
Liu~\etal~\cite{liu_etal_generalized_multiway_cut_in_trees_2012} and Kanj~\etal~\cite{kanj_etal_improved_algorithms_cut_problems_on_trees_2015}.


\myitparagraph{Min-max cost.}
The \emph{min-max} cost of a multicut is the maximum weight of the boundaries of $\multicutM$'s components,
where a boundary of a component is the set of edges leaving this component. 
By $\ConMulticutMaxTrees$ we denote the problem of optimizing the min-max objective for connected multicuts in trees.
It was first considered by Svitkina and Tardos~\cite{svitkina_tardos_min-max_multiway_cuts_2004}, who
observed that the connectivity requirement in $\ConMulticutMaxTrees$ is essential (see also Appendix~\ref{sec: connectivity requirement for min-max}).
A systematic study of  $\ConMulticutMaxTrees$ was initiated by Kolman and Tiwary~\cite{kolman_tiwary_2026},
who focussed on its multiway-cut variant $\ConMultiwayCutMaxTrees$ (where $\ConflictGraph$ is a clique) and
proved that it is fixed-parameter tractable when parametrized by the number of terminals,
and showed that its natural polytope has superpolynomial extension complexity.


\myparagraph{Our results.}
In the first part of this paper, in Section~\ref{sec: 2-approximation minmax}, we focus on algorithms for $\ConMulticutMaxTrees$
that apply to arbitrary conflict graphs. Unlike $\MulticutSumTrees$, that has substantial literature, 
the only known complexity results for $\ConMulticutMaxTrees$ are those in~\cite{kolman_tiwary_2026}, that 
only address its multiway-cut restriction $\ConMultiwayCutMaxTrees$ (when $\ConflictGraph$ is a clique).
By the same argument as for $\MulticutSumTrees$, $\ConMulticutMaxTrees$ cannot be approximated with ratio smaller than $2$.
The idea behind the LP-formulation for $\MulticutSumTrees$ in~\cite{garg_etal_primal-dual_multicut_trees_1997}
does not apply to $\ConMulticutMaxTrees$, and,
in fact, $\ConMulticutMaxTrees$ does not seem to have a simple relaxation useful for bounding the optimal min-max cost.
Nevertheless, we formulate an auxiliary optimization problem called the \emph{local cut} that does have a simple linear
program relaxation  and show that $\ConMulticutMaxTrees$ can be reduced to the problem of
computing local cuts, yielding an algorithm with an optimal approximation ratio.


\begin{theorem}\label{thm: approximation for min-max variant}
(a) For each constant $\epsilon > 0$, there is a polynomial-time $(2+\epsilon)$-approximation algorithm for
{\normalfont{$\ConMulticutMaxTrees$}}.
(b) There is a polynomial-time $2$-approximation algorithm for the unweighted variant of {\normalfont{$\ConMulticutMaxTrees$}}.
\end{theorem}

Somewhat surprisingly, the approximation algorithm above can be modified to yield
a polynomial-time Turing reduction from unweighted $\ConMulticutMaxTrees$ to $\VertexCover$, which is simultaneously
approximation-preserving and parameter-preserving with respect to the objective functions.
In particular, by applying known vertex cover algorithms \cite{harris_narayanaswamy_faster_algorithm_vertex_cover_parameterized_by_solution_size_2024,clinch_etal_faster_randomized_algorithm_vertex_cover_automated_approach_2025,brankovic_fernau_novel_parameterized_approximation_vertex_cover_2013}, we obtain mildly-exponential algorithms for the unweighted case.


\begin{theorem}\label{thm: paramterized for unweighted min-max}
Consider the variant of {\normalfont{$\ConMulticutMaxTrees$}}, where the input instances consist of
an unweighted tree $T$, a conflict graph $\ConflictGraph$, and a cost threshold~$k$.
\begin{description}
	\item{(a)} 
	There are algorithms with running times $O^*(1.2529^k)$ and $O^*(1.1996^{|\Gamma|})$ that determine whether a multicut 
	of cost at most $k$ exists\footnote{$O^*()$ hides polynomial factors.}.
	\item{(b)} 
	There is an algorithm with running time $O^*(1.09^k)$ that
	either computes a multicut with cost at most $1.5k$ or determines that every multicut has cost greater than $k$.
\end{description}
\end{theorem}

In contrast, the fastest known algorithm for constructing an optimal solution to unweighted $\MulticutSumTrees$ runs in 
$O^*(1.554^k)$ time~\cite{kanj_etal_improved_algorithms_cut_problems_on_trees_2015,chen_etal_multicut_in_trees_vertex_cover_2011}, where $k$ is the min-sum objective.

\smallskip


The second part of the paper, Sections~\ref{sec: polynomial time algorithm multiway minmax}~and~\ref{sec: neighborhood diversity algorithms}, 
studies exact algorithms for both cost versions, with running times dependent on the structure of the conflict graph.
We start, in Section~\ref{sec: polynomial time algorithm multiway minmax},
by presenting a polynomial-time algorithm for $\ConMultiwayCutMaxTrees$, the restriction of $\ConMulticutMaxTrees$ to
conflict graphs that are cliques.


\begin{theorem}\label{thm: poly algorithm for min-max for trees}
\normalfont{$\ConMultiwayCutMaxTrees$} can be solved in $O(n \log \totalweight)$ time for integer weights 
(where $\totalweight$ is the total edge weight) and $O(n^2)$ time for real weights.
\end{theorem}

This significantly improves the previous results~\cite{kolman_tiwary_2026}, where the solution
involved a parametrization by the number of terminals.
Our algorithm is substantially more involved than the polynomial-time algorithm for the min-sum 
analog~\cite{chopra_rao_multiway_cut_polyhedron_1991,dahlhaus_1994_complexity_multiterminal_cuts,hasan_motwani_query_optimization_1995,costa_billionnet_multiway_cut_trees_2004}.
In essence, we show that in polynomial time we can verify, using dynamic programming, if there exists a solution with cost not exceeding a
given value $L$, so for integer weights the optimum can be found in polynomial time. 
To extend it to real weights, the dynamic programming algorithm processes \emph{whole cost functions}, which we show 
to be piece-wise constant with only a linear number of discontinuities.
In addition to its own interest, this algorithm also serves as a stepping stone towards
much more general results that we discuss later in the paper.

These more general results are obtained via parametrizations.
The hardness result of Garg et al.~\cite{garg_etal_primal-dual_multicut_trees_1997} applies even
when the input tree $T$ is a star, or when $T$ has depth 2 and $\ConflictGraph$ is a 
matching, so fixing some commonly used structural parameters of the input graph (say, maximum vertex degree, 
modular-width, etc.) cannot produce efficient algorithms for $\MulticutSumTrees$.
Parametrizations by the number of edges in the optimum solution produce running times that are overly pessimistic
for some simple instances, for example when the confict graph is a clique.
This raises a natural question: is there a parametrization \emph{based only on the conflict graph structure}
that includes as special cases the known variants of
$\MulticutSumTrees$ with polynomial-time algorithms, namely the
instances when the conflict graph is a complete multi-partite graph~\cite{hasan_motwani_query_optimization_1995}
or a union of a constant number of (possibly overlapping) cliques~\cite{liu_etal_generalized_multiway_cut_in_trees_2012,kanj_etal_improved_algorithms_cut_problems_on_trees_2015}. 
The same question applies to $\ConMulticutMaxTrees$.

One simple parameter that naturally captures the relevant structure of confict graphs is called
\emph{neighborhood-diversity} (see~\cite{lampis_algorithmic_meta_theorems_for_restrictions_of_treewidth_2012}), denoted here by $\nd(\ConflictGraph)$,
which groups vertices into ``types'' according to their neighborhoods. 
In Section~\ref{sec: neighborhood diversity algorithms}, we give polynomial-time
algorithms for conflict graphs with constant $\nd(\ConflictGraph)$.


\begin{theorem}\label{thm: neighborhood diversity result}
\begin{description}
\item{(a)} {\normalfont{$\MulticutSumTrees$}} can be solved in time $O(n+|\ConflictEdges|+4^{\nd(\ConflictGraph)}|\Terminals|)$.
\item{(b)} {\normalfont{$\ConMulticutMaxTrees$}} can be solved in time $O(n+|\ConflictEdges|+4^{\nd(\ConflictGraph)}|\Terminals|^2)$ for 
real weights and in time $O(n+|\ConflictEdges|+4^{\nd(\ConflictGraph)}|\Terminals|\log W)$ for integer weights,
where $\totalweight$ is the total edge weight.
\end{description}
\end{theorem}

For complete $c$-partite graphs and for $q$-clique covers,
Theorem~\ref{thm: neighborhood diversity result} already matches the running times 
in~\cite{hasan_motwani_query_optimization_1995,liu_etal_generalized_multiway_cut_in_trees_2012,kanj_etal_improved_algorithms_cut_problems_on_trees_2015}
if the parameters $c$ and $q$ are constant. 

To further improve Theorem~\ref{thm: neighborhood diversity result}, we introduce a
generalized notion of types and a new ``split-complexity'' parameter $\spcomp(\TypeGraph)$, defined on a ``generalized type graph'' $\TypeGraph$ associated with $\ConflictGraph$.
Roughly, $\spcomp(\TypeGraph)$ measures the number of ways one can divide an independent set in $\ConflictGraph$ into two parts knowing only which types it contains.
We then show (Section~\ref{sec: generalized type graphs})
that for both cost versions the problem is tractable for polynomially-bounded $\spcomp(\TypeGraph)$.


\begin{theorem}\label{thm: generalized neighborhood diversity result}
\begin{description}
Given a generalized type graph $\TypeGraph$ for the conflict graph $\ConflictGraph$:
\item{(a)} {\normalfont{$\MulticutSumTrees$}} can be solved in time $O(n+\spcomp(\TypeGraph) |\Terminals|)$.
\item{(b)} {\normalfont{$\ConMulticutMaxTrees$}} can be solved in time $O(n+\spcomp(\TypeGraph) |\Terminals|^2)$ for 
real weights and in time $O(n+\spcomp(\TypeGraph) |\Terminals| \log W)$ for integer weights,
where $\totalweight$ is the total edge weight.
\end{description}
\end{theorem}

Applying appropriate bounds to $\spcomp(\TypeGraph)$, we get algorithms for instances
when $\ConflictGraph$ has bounded neighborhood diversity (covered by Theorem~\ref{thm: neighborhood diversity result}),
is a union of cliques, and is a complete multipartite graph.


\begin{corollary}\label{cor: neighborhood diversity + clique cover + multipartite result} 
\begin{description}
\item{(a)} If $\ConflictGraph$ is a complete $c$-partite graph, then {\normalfont{$\MulticutSumTrees$}} can be solved in time $O(n+c|\Terminals|)$. 
Likewise {\normalfont{$\ConMulticutMaxTrees$}} can be solved
in time $O(n+c |\Terminals|^2)$ time for real weights and $O(n+c|\Terminals| \log \totalweight)$ for integer weights.
\item{(b)} If $\ConflictGraph$ is a union of $q$ cliques, then {\normalfont{$\MulticutSumTrees$}} can be solved in time $O(n+3^q|\Terminals|)$.
Likewise {\normalfont{$\ConMulticutMaxTrees$}} can be solved
in time $O(n+ 3^q |\Terminals|^2)$ time for real weights and $O(n+3^q|\Terminals| \log \totalweight)$ for integer weights.
\end{description}
\end{corollary}

Previously, Corollary~\ref{cor: neighborhood diversity + clique cover + multipartite result} was only known for $\MulticutSumTrees$ \cite{hasan_motwani_query_optimization_1995,liu_etal_generalized_multiway_cut_in_trees_2012}.
Note that in Corollary~\ref{cor: neighborhood diversity + clique cover + multipartite result}(b) we assume that the clique cover is given on input.
This assumption is necessary, as the best algorithm finding a $q$-clique cover takes 
$O^* (2^{2^q})$ time~\cite{gramm_etal_data_reduction_and_exact_algorithms_for_clique_cover_2009} and is optimal assuming ETH~\cite{cygan_etal_clique_cover_2016}. 
For this reason, we also have to assume that the object $\TypeGraph$ is given for
Theorem~\ref{thm: generalized neighborhood diversity result} in general.\footnote{
Technically, Corollary~\ref{cor: neighborhood diversity + clique cover + multipartite result}(a) also assumes the parts are already given, but computing that can be done in $O(|\Terminals|+|\ConflictEdges|)$ time.}
But for Theorem~\ref{thm: neighborhood diversity result}, no additional information is needed.

Due to space constraints, some proofs are included in the appendix.
We also include in the appendix some side results that refine the computational complexity picture of
problems related to connected multicuts. 
Svitkina and Tardos~\cite{svitkina_tardos_min-max_multiway_cuts_2004} considered a bi-objective version
of connected multicuts in trees that involves both min-sum and min-max objective functions.
In Appendix~\ref{sec: np-hardness tree cutting} we show that that problem is (weakly) $\NP$-hard even if the conflict graph is a clique.
In Appendix~\ref{sec: connectivity requirement for min-max} we show that the \emph{connectivity} requirement 
is essential for the min-max variants of $\MultiwayCutMax$, proving the ratio of $\Omega(\log n)$
in the case of trees and $\Omega(\sqrt{n})$ for arbitrary graphs, between the optimal solutions without and with this requirement.
This strengthens the observation from~\cite{svitkina_tardos_min-max_multiway_cuts_2004} showing that this ratio is at least $1.25$.


\myparagraph{Other related work.}
There is a variety of other variants of $\MulticutSumTrees$ in the literature. For example, in some models not all conflicting pairs need to be separated. 
The \emph{partial multicut problem} asks for a set of edges which separate at least a given number of conflicting pairs, 
while the \emph{prize-collecting multicut problem} adds a penalty for non-separated conflicts. Both problems admit 
constant-ratio approximations~\cite{levin_segev_partial_multicuts_in_trees_2006,mestre_lagrangian_relaxation_and_partial_cover_2008}.
Further studied extensions involve submodular cost functions on edge sets or
submodular penalties for non-separated conflicts~\cite{liu_li_multicut_problem_in_trees_with_submodular_penalties_2019,
hou_liu_hou_prize_collecting_multicut_on_trees_2020,
hou_et_al_submodular_multicut_in_trees_with_linear_penalties_2021,
liu_li_prize_collecting_multicut_problem_in_trees_2022,
liu_li_prize_collecting_multicut_in_trees_with_submodular_penalties_2024}.
Multicuts for directed trees were studied as well~\cite{costa_et_al_greedy_algorithm_for_multicut_and_integral_multiflow_in_rooted_trees_2003}.

Some of the results for trees have been extended to bounded-treewidth graphs, see for example, the paper by Bentz and Le Bodic~\cite{bentz_lebodic_multicut_problem_bounded_treewidth_2020}.

The multicut problem  for arbitrary graphs with the min-sum objective, $\MulticutSum$, has rich literature.
It's known to be computationally hard. Even its multiway-cut version $\MultiwayCutSum$ with only three terminals is $\NP$-hard~\cite{dahlhaus_1994_complexity_multiterminal_cuts}.
Research on $\MulticutSum$ includes studies of 
its LP-relaxations~\cite{cherkasskij_solution_of_a_problem_on_multicommodity_flows_in_a_network_1977,bertsimas_etal_lp_relaxations_multicut_1999,chopra_rao_multiway_cut_polyhedron_1991},
estimates of the optimum in terms of maximal flows~\cite{garg_vazirani_yannakakis_approximate_max_flow_min_multicut_theorems_1996},
and
approximation algorithms~\cite{calinescu_karloff_rabani_improved_approximation_algorithm_for_multiway_cut_2000,karger_et_al_rounding_algorithms_for_geometric_embedding_of_minimum_multiway_cut_2004}.
For more information, see the survey by Costa et al.~\cite{costa_et_al_minimal_multicut_and_maximal_integer_multiflow_survey_2005}. 

The min-max variant is less understood. Its multiway-cut subproblem, $\MultiwayCutMax$, was studied by
Svitkina and Tardos~\cite{svitkina_tardos_min-max_multiway_cuts_2004}, who
showed that it is $\NP$-hard even for only four terminals,
and developed polylogarithmic approximation algorithms, later
improved by Bansal~\etal~\cite{bansal_etal_min-max_partitioning_2014}. 
Svitkina and Tardos also considered the subproblem $\MultiwayCutMaxTrees$, the restriction of
$\MultiwayCutMax$ to trees (without the connectivity requirement), for which they proved $\NP$-hardness and
gave a $(2+\epsilon)$-approximation algorithm for any $\epsilon > 0$.

Kolman and Tiwary~\cite{kolman_tiwary_2026} studied the variant of $\MultiwayCutMax$ with the connectivity requirement on multicut's components.
Their results include several hardness results, for example that $\ConMultiwayCutMax$ is (weakly) $\NP$-hard even for three terminals and graphs
of treewidth three, a pseudo-polynomial algorithm, and an FPTAS for graphs of bounded treewidth.


\section{Preliminaries}
\label{sec: preliminaries}
We are given a rooted tree $T =(V,E,\weight)$  with a weight function $\weight: E\to \reals^+$
and a graph $\ConflictGraph =(\Terminals, \ConflictEdges)$, where $\Terminals \subseteq V$. The vertices in $\Terminals$ are called \emph{terminals},
the edges of $\ConflictEdges$ between them are referred to as \emph{conflicts}, and $\ConflictGraph$ is called the \emph{conflict graph}.
Throughout the paper, the number of vertices in $T$ is denoted
by $n=|V|$, and by $\maxweight = \max_{e\in E}\weight(e)$ and $\totalweight = \sum_{e\in E}\weight(e)$
we denote the maximum and total edge weight, respectively. 
We use the notation $\pi_v$ for the parent of $v\in V$ (if $v$ is not the root) and $\children(v)$ for the set of children of $v$.
A \emph{connected multicut} of the instance $(T,\ConflictGraph)$ is defined as a partition $\multicutM$ of $V$ 
into subsets called \emph{components}, where each component $X$ satisfies three conditions:
(i) $X \cap \Terminals\ne\emptyset$, 
(ii) $X \cap \Terminals$ is an independent set in $\ConflictGraph$, and
(iii) the subgraph of $T$ induced by $X$ is connected (that is, a subtree of $T$).
For $v\in V$, $\multicutM(v)$ is the component of $\multicutM$ containing $v$.
For $X\subseteq V$ we define the \emph{boundary of $X$} as $\boundary X =E(X,V-X)$; that is $\boundary X$ is the set of edges with exactly one endpoint in $X$.
(In particular, $\boundary\multicutM (v)$ refers to the boundary of $\multicutM(v) \in \multicutM$.)
The boundary of a multicut $\multicutM$ is the union of boundaries of its components:  $\boundary \multicutM = \bigcup_{X\in\multicutM} \boundary X$.
By the connectivity requirement, $\multicutM$ is uniquely determined by its boundary; in particular
$\multicutM$ consists of the connected components of $T-\boundary \multicutM$.
We consider two cost functions for multicuts:
the \emph{sum-cost} of $\multicutM$ is the total weight of the edges in its boundary, 
$\sumcost(\multicutM) = \weight(\boundary \multicutM) = \half \sum_{X \in \multicutM} \weight(\boundary X)$,
and its  \emph{max-cost}  is the maximum cost of its component boundaries,
$\maxcost(\multicutM) = \max_{X\in \multicutM} \weight(\boundary X)$.
(See Figure~\ref{fig: multicut example} for an example.)


\begin{figure}[t]
\begin{center}
\includegraphics[width = 4in]{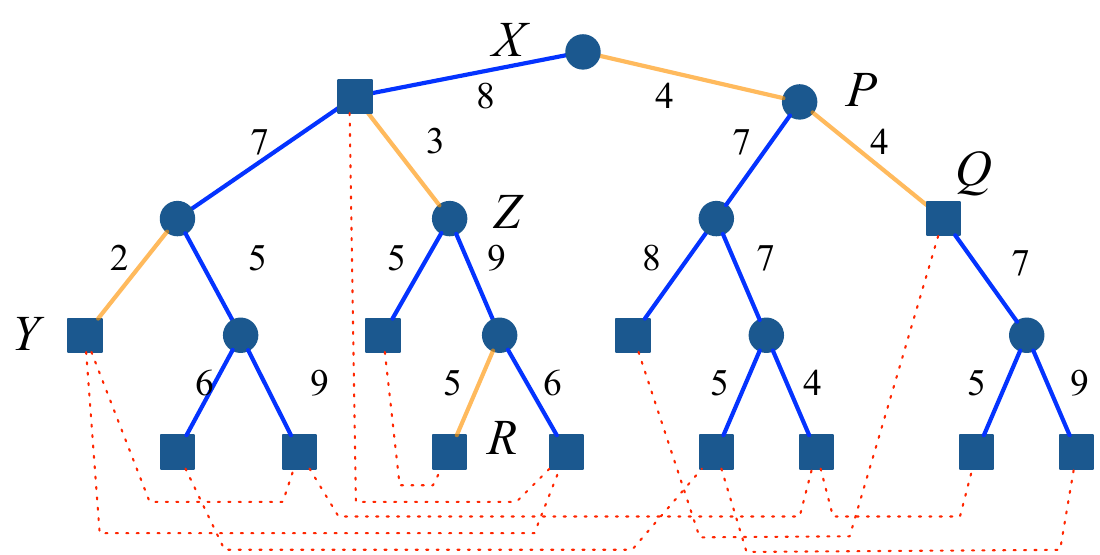}
\caption{%
A weighted tree, a conflict graph, and
 a connected multicut $\multicutM = \braced{X,Y,Z,P,Q,R}$. 
Tree edges are represented by solid lines, non-terminals by circles, terminals by squares, and
conflicts by dotted thin lines.
The edges within components are in blue (dark), 
the boundary edges are in yellow.  Each component's label is near its top vertex.
The boundary costs are $\weight(\cutof X) = 9$, $\weight(\cutof Y) = 2$, 
$\weight(\cutof Z) = 8$, $\weight(\cutof P) = 8$, $\weight(\cutof Q) = 4$, and $\weight(\cutof R) = 5$. 
So $\sumcost(\multicutM) = 18$ and $\maxcost(\multicutM) = 9$.}
\label{fig: multicut example}
\end{center}
\end{figure}


We focus on two optimization problems for trees, the min-sum and min-max variants of computing optimal connected multicuts:
\begin{description}
\item{$\MulticutSumTrees$:} Given a tree $T$ and a conflict graph $\ConflictGraph$, compute a connected multicut $\multicutM$ that minimizes its sum-cost $\sumcost(\multicutM)$.
\item{$\ConMulticutMaxTrees$:} Given a tree $T$ and a conflict graph $\ConflictGraph$, compute a connected multicut $\multicutM$ that minimizes its max-cost $\maxcost(\multicutM)$.
\end{description}
We denote the optimal cost of each solution by $\sumopt(T,\ConflictGraph)$ and $\maxopt(T,\ConflictGraph)$, respectively.
In the $\MulticutSumTrees$ notation we do not indicate the connectivity requirement, because, as mentioned in the introduction, it's redundant: any multicut satisfying
all other requirements can be converted into a connected one without increasing its sum-cost. This requirement is essential, however, for $\ConMulticutMaxTrees$
(see Appendix~\ref{sec: connectivity requirement for min-max}).


\myparagraph{Normal forms.}
We assume that all input trees are rooted. For a rooted tree $T$ and a vertex $v$, by $T_v$ we will denote the
subtree of $T$ rooted at $v$ and containing all $v$'s descendants.

Our algorithms will first simplify the input trees to streamline the computation.
We say that $T$ is in \emph{normal form} if all its leaves are terminals and all internal non-terminals have at least two children.
If we further require that all internal nodes are non-terminals (so $\Terminals$ is exactly the set of leaves)
and have exactly two children then we say that $T$ is in \emph{binary normal form}
(See Figure~\ref{fig: multicut example}.)
It is known, and not difficult to prove that, after $O(n)$ pre-processing, without any loss of generality we can assume that
a weighted tree is in binary normal form, and that an unweighted binary tree is in a normal form.
This is formalized in the lemma below. For the sake of completeteness, we include its proof in Appendix~\ref{sec: normal forms proof}.


\begin{figure}[t]
\begin{center}
\includegraphics[width = 4in]{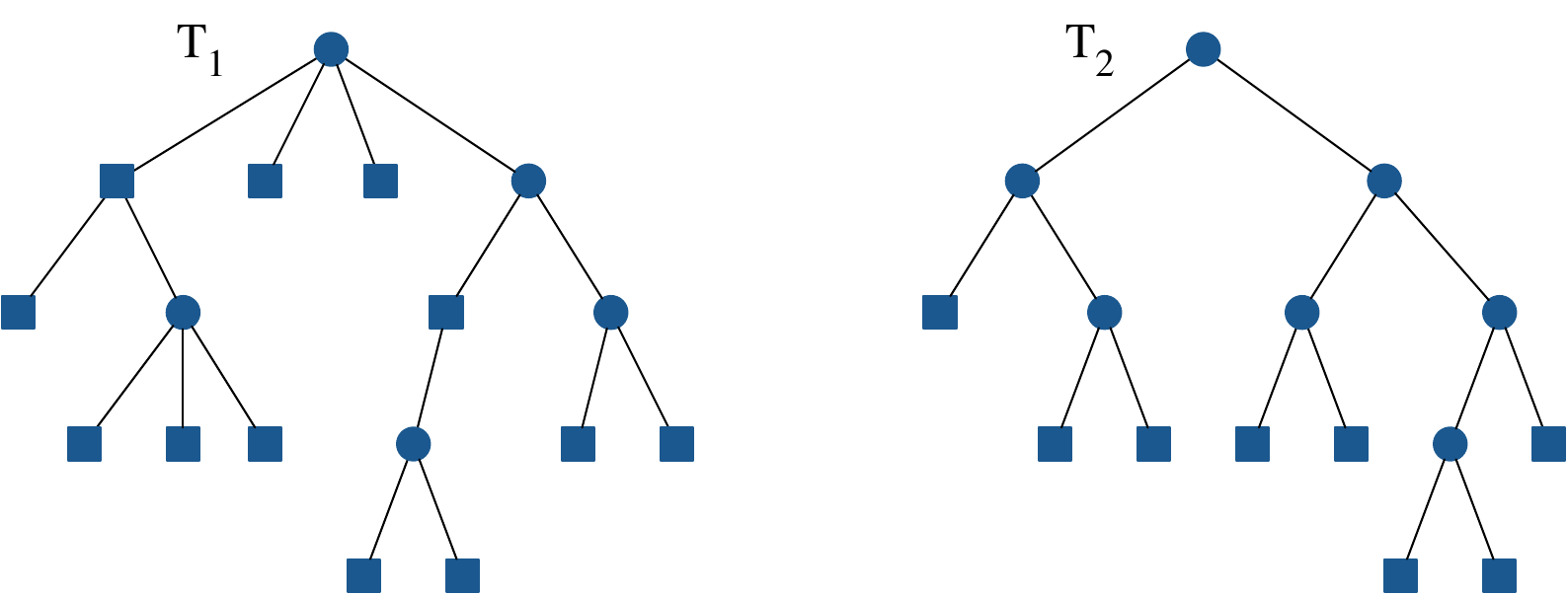}
\caption{%
On the left, tree $T_1$ in normal form. On the right, tree $T_2$ in binary normal form.
Terminals are represented by squares and non-terminals by circles.}
\label{fig: normal forms}
\end{center}
\end{figure}


\begin{lemma}\label{lemma: normal form for trees}
Let $T=(V,E,\weight)$ be a weighted tree and $\ConflictGraph=(\Terminals, \ConflictEdges)$ a conflict graph.
In time $O(n+|\calF|)$ one can compute a tree  $T'=(V',E',\weight')$ with $\Gamma \subseteq V'$ such that 
$\sumopt(T',\ConflictGraph)=\sumopt(T,\ConflictGraph)$, $\maxopt(T',\ConflictGraph)=\maxopt(T,\ConflictGraph)$, and
\begin{description} 
    \item{(a)} If $T$ is unweighted, then $T'$ is an unweighted tree in normal form,
	\item{(b)} If $T$ is weighted, then $T'$ is in binary normal form and $\max_{e \in E'} \weight'(e) \le \sum_{e \in E} \weight(e)$.
\end{description}
In both cases and for both objective functions, any multicut $\multicutM'$ for $T'$ can be converted into a multicut $\multicutM$ for $T$ in time $O(n)$, without increasing the cost.
\end{lemma}

A tree in normal form has $O(|\Gamma|)$ vertices, since non-terminals (except the root) have degree at least 3. 
Thus we may use $|\Gamma|$ and $n$ interchangeably in the asymptotic formulas.



\section{Algorithms for $\ConMulticutMaxTrees$}
\label{sec: 2-approximation minmax}


In this section we study algorithms for $\ConMulticutMaxTrees$, that is computing
a connected multicut in trees with the min-max cost function, and
we prove Theorems~\ref{thm: approximation for min-max variant} and~\ref{thm: paramterized for unweighted min-max}.
To develop our algorithms, we proceed in two steps.
In Section~\ref{subsec: local cut} we first identify a simpler optimization problem where the goal is
only to compute one ``local cut''. 
This problem, while still $\NP$-hard itself, has a simple LP relaxation that can be used to produce
a $2$-approximate solution.
Then, in Section~\ref{subsec: reduce min max to local cut}, we give a polynomial-time
Turing reduction from $\ConMulticutMaxTrees$ to this problem using dynamic programming.
This will simultaneously prove Theorems~\ref{thm: approximation for min-max variant}~and~\ref{thm: paramterized for unweighted min-max}.


\myitparagraph{Intuitions.}
A natural idea for solving $\ConMulticutMaxTrees$
would be to compute bottom-up a multicut for each rooted sub-tree: 
to compute a multicut $D(v)$ of a subtree $T_v$,
first include in it the computed edges $D(u)$ for each $u\in ch(v)$, 
and then add some suitably chosen edges between $v$ and its children to eliminate the conflicts between terminals in subtrees of different children of $v$.
This works for unweighted trees (see Appendix~\ref{sec: faster unweighted approximation}) but fails for the weighted case,
simply because the edges from $v$ to its children could be very expensive. 

This observation leads to a somewhat unusual algorithm that, when processing a node $v$, computes a multicut for $T_v$
using dynamic-programming entries associated with all \emph{descendants} of $v$, rather than only with its children.
Roughly, the information computed for $v$'s descendants will be used to eliminate edges that are too expensive to
be in the boundary of $v$'s component. After excluding these edges (by simply contracting them),
the component of $v$ is then determined by computing its local cut. This is the essence of the
Turing reduction mentioned in the previous paragraph.


\myitparagraph{Reductions.}
There are in fact a few Turing reductions in this section. To reduce a minimization problem $\problemP$ to a minimization problem $\problemP'$,
given a parameter $\rho\ge 1$, the reductions we use have the following form:

\begin{itemize}[nosep]
\item We assume the existence of an oracle for $\problemP'$ that, given an instance $I'$ of $\problemP'$ and a value $L'$, 
either returns a solution of cost at most $\rho L'$ or determines that the optimum cost of $I'$ is larger than $L'$.

\item
The algorithm for $\problemP$ uses the oracle for $\problemP'$ and,
given an instance $I$ of $\problemP$ and a value $L$, it
either returns a solution of cost at most $\rho L$ or determines that the optimum cost of $I$ is larger than $L$.
Furthermore, it runs in polynomial time, and each oracle query $(I', L')$ satisfies $L' \le L$.
\end{itemize}
Note that we allow $\rho=1$ (in which case we just have a parameter-preserving reduction), and that 
$L$ and $L'$ may be real numbers. When weights are real, we only obtain a $(\rho+\epsilon)$-approximation for $\problemP$ from a $\rho$-approximation of $\problemP'$.


\subsection{Computing Local~Cuts} \label{subsec: local cut}

Given a tree $T = (V,E,\weight)$, a conflict graph $\ConflictGraph$, and a vertex $\starv \in V$, we define a \emph{local cut for $\starv$} to be a subset $X\subseteq V$
such that (i) $\starv \in X$, (ii) $T[X]$ is connected (i.e., a subtree of~$T$), and (iii) $X \cap \Terminals$ is an independent set in $\ConflictGraph$ (possibly empty). 
$X$ is a \emph{proper local cut} if $X \cap \Terminals \ne \emptyset$, and \emph{improper} otherwise.
We now define the local cut problem:
\begin{description}
\item{$\LocalCutTrees$:} Given a tree $T = (V,E,\weight)$, a conflict graph $\ConflictGraph$, and  $\starv \in V$, 
compute a proper local cut $X \subseteq V$ for $\starv$ minimizing $\weight(\cutof X)$.
\end{description}

We let $\LocalCutTreesImprop$ be the problem variant where improper cuts are allowed, and we denote the optimum cost for both problems by 
$\locopt(T, \ConflictGraph, \starv)$ and $\locoptim(T, \ConflictGraph, \starv)$, respectively.
First, we observe that $\LocalCutTrees$ is not harder than $\LocalCutTreesImprop$.


\begin{observation}\label{obs: reduce proper local cuts to improper cuts}
There is a polynomial-time Turing reduction from $\LocalCutTrees$ 
to $\LocalCutTreesImprop$ that makes $O(n)$ oracle queries on an input instance with \(n\) vertices. 
Moreover, if the input instance is unweighted, then every instance submitted to the oracle is unweighted.
\end{observation}

\begin{proof}
We actually prove a stronger fact: given an (approximate) algorithm for $\LocalCutTreesImprop$ and a vertex $u\in V$, we compute an (approximate) minimum weight 
local cut $X$ for $(T, \ConflictGraph, \starv)$ that contains $u$ (if one exists).
We then try this procedure for each $u\in\Terminals$ and pick the solution of minimum weight.
	
Given $u$, we define $(T', \ConflictGraph', \starv)$ by contracting the unique path in $T$ between $u$ and $\starv$ into a single vertex, 
then taking a graph quotient in $\ConflictGraph$ by merging all the terminals (if any) that appear on this path.
(If two such terminals have a conflict in $\ConflictGraph$, then there is no local cut for $\starv$ containing $u$.) 
We have an obvious correspondence between local cuts of $T'$ and local cuts of $T$ containing $u$, and clearly $T'$ is unweighted if $T$ is.
\end{proof}

Next, we give a $2$-approximation algorithm for computing local cuts.
\begin{lemma}\label{lem: two approximation for local cut}
There is a polynomial-time $2$-approximation algorithm for $\LocalCutTreesImprop$ 
(and thus also for $\LocalCutTrees$, by Observation~\ref{obs: reduce proper local cuts to improper cuts}).
\end{lemma}

\begin{proof}
Given an instance $(T, \ConflictGraph, \starv)$ with $T=(V,E,\weight)$, $\ConflictGraph=(\Terminals,\ConflictEdges)$, and $\starv\in V$, we formulate an
LP relaxation with a variable $x_e$ for each edge $e \in E$,
where $x_e = 1$ is intended to indicate that $e$ lies in the
boundary of a local cut.
Letting $P_{uv}$ denote the path in $T$ between $u,v\in V$ (as an edge set), the relaxation is as follows:
\begin{align}
	 	\text{minimize } \sum_{e \in E} \weight(e) &x_e \;
		&\\\
	 	\text{subject to } \sum_{e \in P_{s\starv} } &x_e + \sum_{e \in P_{t\starv}} x_e \;\ge\; 1 \quad \forall st \in \ConflictEdges 
		\label{eq: lp constraint}
		\\
		 &x_e \;\ge\; 0 \quad \forall e \in E
		 \nonumber
\end{align}
We now describe a rounding procedure of a solution $x$ for this linear program. Let $L$ denote the optimum solution value.
For $\rho \in (0,1)$ let $B_\rho = \braced{u \in V \mid x(P_{u\starv}) \le \rho}$, where $x(F) = \sum_{e \in F} x_e$ for a set $F\subseteq E$;
clearly, $T[B_\rho]$ is connected and $\starv \in B_\rho$.
We aim to show that for a proper choice of $\rho$, $B_\rho$ is a local cut for $\starv$ of cost at most $2L$.

We claim that $\int_0^1 \weight(\cutof B_{\rho}) \, d\rho \le L$. 
As $|x(P_{u\starv})-x(P_{v\starv})| = x_{uv}$ for every $uv\in E$, we obtain
\begin{align*}
\int_0^1 \weight(\cutof B_{\rho}) \, d\rho \;&=\; \int_0^1 \sum_{e \in \cutof B_{\rho}} \weight(e) \, d\rho 
									\;=\; \sum_{uv \in E} \weight(uv) \int_0^1 1_{uv \in \cutof B_{\rho}} \, d\rho \\
									\;&\le\; \sum_{uv \in E} \weight(uv) |x(P_{u\starv})-x(P_{v\starv})| \;=\; \sum_{uv \in E} \weight(uv) x_{uv} =L
\end{align*}
where $1_{\Pi}$ is an indicator function for a predicate $\Pi$.

By the above claim, there exists $\starrho \in (0,\half)$ such that $\weight(\cutof B_{\starrho})\le 2L$. 
To finish the proof, it suffices to verify that 
$B_{\starrho}\cap \Terminals$ is an independent set in $\ConflictGraph$.
For any distinct $s,t\in B_{\starrho} \cap \Terminals$, we have
$x(P_{s\starv}) + x(P_{t\starv}) \le 2 \starrho < 1$,
so $st \notin \ConflictEdges$ by the LP constraint~(\ref{eq: lp constraint}).
Thus, $B_{\starrho}$ is the desired local cut for $\starv$.
As there are at most $n$ distinct sets $B_\rho$ for
$\rho \in (0,\half)$, namely those for values of $\rho$ of the form
$\rho = x(P_{v\starv})$ for $v \in V$, the radius $\starrho$
and the corresponding local cut $B_{\starrho}$ can be computed in polynomial time.
\end{proof}

There is an alternative way to prove Lemma~\ref{lem: two approximation for local cut}, via a primal-dual
algorithm instead of randomized rounding. We include it in Appendix~\ref{sec: primal dual algorithm}.
That solution does not require solving a linear program.

We make one last observation about $\LocalCutTreesImprop$ on unweighted trees. As mentioned earlier, 
a reduction from $\VertexCover$ shows that $\MulticutSumTrees$ and $\ConMulticutMaxTrees$ have no $(2-\epsilon)$-approximation under standard assumptions,
even on unweighted stars. What's interesting about $\LocalCutTreesImprop$ is that it is actually equivalent to $\VertexCover$ when $T$ is unweighted.
For a graph $H$, let $\vcover(H)$ denote the minimum size of a vertex cover of $H$.


\begin{observation}\label{obs: reduction from unweighted local cut to vertex cover}
Given an unweighted $\LocalCutTreesImprop$ instance $(T,\ConflictGraph, \starv)$, we can in linear time compute a graph $H$ and number $k'$ where 
$ \vcover(H) +k' =\locoptim(T,\ConflictGraph, \starv)$. 
Also, any vertex cover $A'\subseteq V(H)$ can be converted in linear time to a local cut of $\starv$ with cost $|A'|+k'$.
\end{observation}

\begin{proof}
Root $T$ at $\starv$. Let $A\subseteq ch(\starv)$ be the set of children $u$ such that 
there are no conflicts between the terminals in $T_u\cup\braced{\starv}$.
Let $B= ch(\starv) - A$. Define $H=(A, E_H)$ by adding $uw$ to $E_H$ if there exist terminals $s\in T_u$ and $t\in T_w$ such that $st\in\ConflictEdges$. ($H$ can easily be computed in linear time.)
	
We now claim that $\vcover(H)+|B|=\locoptim(T, \ConflictGraph, \starv)$. To see this, observe that for any local cut $X\subseteq V$ of $\starv$, $\cutof X$ 
must contain at least one edge from $\braced{u\starv} \cup E(T_u)$ for each $u\in B$, as
otherwise $X$ would have a conflict.
Furthermore, $A'= \braced{u \in A \suchthat |\cutof X \cap (\braced{u\starv} \cup E(T_u))| \ge 1}$ is a vertex cover for $H$.
(If $uw\in E_H$ was not covered, then some conflicting pair in $T_u \cup T_w$ would be in $X$.) This implies $|\cutof X| \ge \vcover(H)+|B|$.
And given a vertex cover $A'$ for $H$,  the set $X=\braced{v^*} \cup \bigcup_{u\in A-A'} T_u$ with boundary
$\cutof X=\braced{u\starv \suchthat u \in A' \cup B}$ is a local cut for $\starv$.
\end{proof}


\subsection{Main Algorithm} \label{subsec: reduce min max to local cut}

In this section we prove Theorem~\ref{thm: approximation for min-max variant}, by
reducing $\ConMulticutMaxTrees$ to $\LocalCutTreesImprop$. 

\begin{lemma}\label{lem: min-max approximation reduction}
There is a polynomial-time Turing reduction from \normalfont{$\ConMulticutMaxTrees$} to \normalfont{$\LocalCutTreesImprop$}.
\end{lemma}
	
To prove Lemma~\ref{lem: min-max approximation reduction}, we assume that there is a constant $\rho$ and an
oracle for $\LocalCutTreesImprop$ which, given an instance  $(T',\ConflictGraph', \starv)$ and a real number $L'$, either returns a solution of cost at most $\rho L'$ 
or determines that $\locoptim(T',\ConflictGraph',\starv)>L'$.
Using this oracle we then describe an algorithm that,
given an instance $(T,\ConflictGraph)$ of $\ConMulticutMaxTrees$ and a real number $L$, 
either returns a multicut for $(T,\ConflictGraph)$ of cost at most $\rho L$, or determines that $\maxopt(T,\ConflictGraph) > L$. 
All queries to the oracle for $\LocalCutTrees$ use values $L' \le L$.

In this section we assume that $T$ is a tree in normal form, rooted at a vertex $r$. 
To streamline the presentation, we add to $T$ a ``dummy'' parent vertex $\pi_r$ with $\weight(r\pi_r) = 0$.


\myitparagraph{Intuitions.}
As a warm-up, we first show a simpler algorithm for $\ConMulticutMaxTrees$ with approximation ratio $8$.
The algorithm proceeds top-down, starting with the root $r$.
Compute a (proper) local cut $X_r\subseteq V$ for $r$ using the algorithm
from Lemma~\ref{lem: two approximation for local cut}. Then proceed recursively:
for each vertex $u$ with $u\pi_u \in \cutof X_r$, compute a local cut $X_u$ in $T_u$ (that is, treating $T_u$ as the input tree)
using Lemma~\ref{lem: two approximation for local cut},
and so on. The vertices $u$ for which local cuts $X_u$ are computed in this process will be called \emph{hubs}.
The computed multicut $\multicutM$ for $T$ is the set of all $X_u$'s for hubs $u$.

Lemma~\ref{lem: two approximation for local cut} implies that for each hub $u$, the cost of the local
cut $X_u$ in $T_u$ satisfies
$\weight(\cutof X_u)\le 2 \max_{v \in V} \locopt(T_v, \ConflictGraph \cap T_v, v)$. 
Every component of $\multicutM$ has a boundary in the form $\cutof X_u\cup\braced{u\pi_u}$, for some hub $u$, 
so its boundary cost is $\weight(\cutof X_u) + \weight(u\pi_u) \le 4\max_{v \in V} \locopt(T_v, \ConflictGraph \cap T_v, v)$.
To show an upper bound of $8$ on the approximation ratio,
it suffices to show that
$\locopt(T_u, \ConflictGraph \cap T_u, u) \le 2 \cdot \maxopt(T, \ConflictGraph)$ for each vertex $u$.
This inequality follows by considering the component $X$ of the optimal multicut for $T$ that contains $u$. Let 
$X' = X \cap V(T_u)$. This $X'$ is a local cut for $u$ in $T_u$, whose boundary within $T_u$
has cost $\weight(\partial_{T_u} X')\le \maxopt(T, \ConflictGraph)$.
If $X'$ is not proper,
we can merge $X'$ with any adjacent component below it (which exists, by the normal form, and must be proper), 
obtaining a proper local cut for $u$ in $T_u$ and at most doubling the cost.

While the bounds in this analysis are tight in the worst case, the extra factors of 2 are really due to a poor choice of components.
If one considers an optimal multicut $\multicutM$ of $T$ with min-max cost $L$, and we take any component $X \in \multicutM$ with common ancestor $u\ne r$, 
then $X$ is a proper local cut for $u$ in $T_u$ with cost at most $L-\weight(u \pi_u)$.
This observation suggests a bottom-up approach: first try to compute min-max multicuts for the sub-trees, 
then compute a local cut for the root using \emph{only} edges $u\pi_u$ for which $T_u$ admits a ``good'' solution.

\medskip

To formalize this idea, we need a slight strengthening of Observation~\ref{obs: reduce proper local cuts to improper cuts}.
For $E'\subseteq E$ and $\starv \in V$, we define
an \emph{$E'$-restricted proper local cut for $\starv$} (or just an \emph{$E'$-restricted cut}, for short)
to be a proper local cut $X$ for $\starv$ such that $\cutof X \subseteq E'$.
Let $\locopt(T, \ConflictGraph, \starv, E')$ denote the minimum weight of all $E'$-restricted proper local cuts for $\starv$.
(If there is no such cut, let $\locopt(T, \ConflictGraph, \starv, E')=\infty$.)

We observe that computing such a cut for $\starv$ can be reduced to $O(n)$ instances of $\LocalCutTreesImprop$, though this is somewhat subtle\footnote{
We could trivially reduce to $\LocalCutTrees$ by increasing the weights of $E-E'$, but this doesn't work if we want an unweighted tree.}.
As in Observation~\ref{obs: reduce proper local cuts to improper cuts}, if $\starv$ is not already a terminal, guess a terminal $t\in\Terminals$ and contract
the edges on the path from $\starv$ to $t$, merging terminals in $\ConflictGraph$ as needed.
If two conflicting terminals were merged, then there is no local cut without separating $\starv$ and $t$.
Otherwise, we identify $\starv$ and $t$ and contract the edges $E-E'$ and merge any terminals in $\ConflictGraph$ if they got merged in $T$.
If we merge two terminals that conflict, this does \emph{not} necessarily mean that there is no solution, but merely that $\starv$ must be separated
from the merged terminal $u$, in this case we simply add $u\starv$ to $\ConflictEdges$. But if we try to merge $\starv$ and $t'$ and $t'\starv\in\ConflictEdges$,
then there is no $E'$-restricted cut without separating $\starv$ and $t$.
If the previous checks pass, then let $(T'_t, \ConflictGraph_t')$ be the modified tree and conflict graph.
Then by inspection, $\locopt(T, \ConflictGraph, \starv, E') = \min_t \locoptim(T'_t, \ConflictGraph'_t, \starv)$, where $t$ ranges over all terminals
where $(T'_t,\ConflictGraph'_t)$ is defined.

We are now ready to describe the reduction, which is simply a dynamic program.


\begin{proof}[Proof of Lemma~\ref{lem: min-max approximation reduction}]
Let $L$ be a real number denoting our target cost and $\rho \ge 1$ be the approximation ratio of our $\LocalCutTreesImprop$ oracle.
From the discussion above, we may assume that this oracle can compute $E'$-restricted proper local cuts:
given a number $L'\le L$ and a tuple $(T, \ConflictGraph, \starv, E')$, where $\starv\in V$ and $E'\subseteq E$,
the oracle either returns an $E'$-restricted cut for $\starv$ of cost at most $\rho L'$ or determines that $\locopt(T, \ConflictGraph, \starv, E') >L'$.

For $v\in V$, define an \emph{$(L,\rho)$-approximate multicut for $T_v$} as a multicut $\multicutQ$ for $(T_v, \ConflictGraph \cap T_v)$ with 
$\maxcost(\multicutQ)\le \rho L$
and with $v$'s component $\multicutQ(v)$ satisfying a stricter requirement
$\weight(\cutof \multicutQ(v)) \le \rho (L - \weight(v \pi _v))$. 
For each $v\in V$ the algorithm will compute a table entry $\dptable_v$
whose value is either an $(L,\rho)$-approximate multicut for $T_v$ or $\nullsol$ (undefined). Let
\begin{align*}
	E_v^1 \;&=\; \braced{u\pi_u \suchthat u \in V(T_v)-\braced{v} \text{ and there is an $(L,1)$-approximate multicut for $T_u$}} 
	\\
	E_v^\rho \;&=\; \braced{u\pi_u \suchthat \text{ $u \in V(T_v)-\braced{v}$ and $\dptable_u$ is an $(L,\rho)$-approximate multicut for $T_u$} }
\end{align*}
We maintain the following invariant for each vertex $v\in V$: 
\begin{description}[nosep]
	\item{$(\star)$} If there is an $(L,1)$-approximate multicut for $T_v$,
then $D_v$ is $(L,\rho)$-approximate.
\end{description}

The validity of $(\star)$ for $v=r$  ensures that the existence of a multicut in $T$ of cost 
at most $L$ implies that $D_r$ is a multicut of cost at most $\rho L$, proving the 
correctness of the reduction.

We now describe an algorithm that guarantees the validity of the invariant $(\star)$ 
at each node $v$. The computation proceeds bottom-up, starting from the leaves.
If $v$ is a leaf, set $\dptable_v=\braced{\braced{v}}$ if $\weight(v\pi_v) \le L$, and 
set $\dptable_v=\nullsol$ otherwise. Then $(\star)$ holds for $v$.

Now suppose that $v$ is an internal vertex and that ($\star$) holds for all $u\in V(T_v)-\braced{v}$.
This implies $E^1_v \subseteq E^\rho_v$.
Using the oracle for $\LocalCutTrees$, we now try to compute an $E^\rho_v$-restricted cut $X_v$ in $T_v$ such that
$\cutof X_v\subseteq E_v^\rho$ and $\weight(\cutof X_v)\le \rho (L- \weight(v\pi_v))$. 
If the oracle fails, set $\dptable_v =\nullsol$; otherwise it returns $X_v$ and set
\begin{equation*}
\dptable_v \;=\; \braced{X_v} \cup \bigcup_{u\pi_u \in \cutof X_v} \dptable_u.
\end{equation*}
We first argue that if the oracle succeeds, then $\dptable_v$ is an $(L,\rho)$-approximate multicut for $T_v$.
Indeed, the component of $D_v$ containing $v$ is $X_v$ and by construction, 
$\weight(\cutof X_v)\le \rho (L- \weight(v\pi_v))$. Consider now a component $X\in D_v$ not containing
$v$, and let $u$ be the common ancestor of its nodes in $X$. 
Then if $u\pi_u \in \cutof X_v$, we have 
$\weight(\cutof{X}) \le \rho (L-\weight(u\pi_u)) + \weight(u\pi_u) \le \rho L$,
and otherwise $\weight(\cutof{X})\le \rho L$ by the inductive assumption.

Therefore, to maintain ($\star$), it suffices to show that if there is an $(L,1)$-approximate 
multicut for $T_v$, then the oracle succeeds.
This is straightforward: suppose we have an $(L,1)$-approximate multicut and let $Y\subseteq T_v$ be its component containing $v$.
By definition of $(L,1)$-approximate multicuts, $Y$ contains a terminal, $\weight(\cutof{Y}) \le L-\weight(v\pi_v)$, and
$\cutof Y\subseteq E_v^1 \subseteq E_v^\rho$.
Thus, $Y$ is an $E_v^\rho$-restricted proper local cut for $v$; in other words $\locopt(T_v, C\cap T_v, v, E_v^\rho) \le \weight(\cutof Y) \le L-\weight(v\pi_v)$.
By the guarantee of our oracle, it succeeds, returning an 
$E_v^\rho$-restricted proper local cut $X_v$ for $v$ with  $\weight(\cutof X_v) \le \rho(L-\weight(v\pi_v))$, as desired.
\end{proof}


\begin{proof}[Proof of Theorem~\ref{thm: approximation for min-max variant}]
Combine Lemma~\ref{lem: two approximation for local cut} and Lemma~\ref{lem: min-max approximation reduction} with binary search.
For the weighted case we need an initial upper bound on the optimum, which can be obtained from the $8$-approximation sketched earlier
or the $2$-approximation for
$\MulticutSumTrees$~\cite{garg_etal_primal-dual_multicut_trees_1997}, which gives a multicut of min-max cost at most $2|\Gamma|$ times the optimum.
\end{proof}


\begin{proof}[Proof of Theorem~\ref{thm: paramterized for unweighted min-max}]
By combining Observation~\ref{obs: reduction from unweighted local cut to vertex cover}, 
Observation~\ref{obs: reduce proper local cuts to improper cuts}, and Lemma~\ref{lem: min-max approximation reduction},
we obtain a polynomial-time reduction from unweighted $\ConMulticutMaxTrees$ to $\VertexCover$ that is approximation- and parameter-preserving, 
in the sense described in the beginning of Section~\ref{sec: 2-approximation minmax}.
The theorem then follows by using the vertex cover algorithms from \cite{harris_narayanaswamy_faster_algorithm_vertex_cover_parameterized_by_solution_size_2024,clinch_etal_faster_randomized_algorithm_vertex_cover_automated_approach_2025,brankovic_fernau_novel_parameterized_approximation_vertex_cover_2013}.
(Regarding the $O^*(1.1996^{|\Gamma|})$ time bound, one can observe that all vertex cover instances computed are 
graph quotients of $\ConflictGraph$ and thus have at most $|\Gamma|$ vertices.)
\end{proof}

A straightforward implementation of 
our 2-approximation algorithm (solving $\LocalCutTreesImprop$ instances with vertex cover) would take time $\tildeO(n^2(n+|\ConflictEdges|))$. 
In Appendix~\ref{sec: faster unweighted approximation} we give an alternative algorithm whose running time is nearly linear.


\section{Polynomial-Time Algorithm for $\ConMultiwayCutMaxTrees$}
\label{sec: polynomial time algorithm multiway minmax}
In this section we prove Theorem~\ref{thm: poly algorithm for min-max for trees}; that is
we give a polynomial-time algorithm for $\ConMultiwayCutMaxTrees$, the 
multiway-cut version of  $\ConMulticutMaxTrees$.
The result in Theorem~\ref{thm: poly algorithm for min-max for trees} is analogous to the result for $\MultiwayCutSumTrees$
in papers~\cite{chopra_rao_multiway_cut_polyhedron_1991,dahlhaus_1994_complexity_multiterminal_cuts,hasan_motwani_query_optimization_1995,costa_billionnet_multiway_cut_trees_2004},
although the min-max cost introduces challenges not present in the min-sum version. 

At a high-level, rather than computing the optimum solution directly, the algorithm
uses an $O(n)$-time dynamic-programming subroutine that, given a real number $L$, either
returns a multiway cut of min-max cost at most $L$ or determines that the optimum min-max cost is greater than $L$.
Using this subroutine, we can compute an optimal multiway cut for integer weights by binary search.
To obtain a strongly polynomial-time algorithm for real weights, we view our dynamic program as computing 
parametric functions over $L$,  which we then show can be represented using only $O(n)$ space per function.


\myparagraph{Notation.} 
We assume that $T$ is in binary normal form with root $r$. 
We also add a dummy parent $\pi_r$ to $r$ and set $\weight(r\pi_r)=0$. 
For a vertex $v\in V$, let $\treeplus_v = T_v \cup \braced{\pi_v}$, that is the sub-tree rooted at $v$ plus $v$'s parent.\footnote{
One can formulate a dynamic program using $T_v$ instead of $\treeplus_v$, but the latter is slightly more convenient.}
A \emph{$\treeplus_v$-cut} is a partition $\multicutM$ of $\treeplus_v$ satisfying the same conditions as a connected multicut,
with the possible exception that the component $\multicutM(\pi_v)$ may not contain a terminal.
We say $\multicutM$ is \emph{proper} if $\multicutM(\pi_v)$ does contain a terminal and \emph{improper} otherwise.
In particular, a proper $\treeplus_r$-cut is just a connected multicut of $T$.
When $v$ is clear from context, we often call $\multicutM(\pi_v)$ the ``top component'' of $\multicutM$.
To avoid clutter, for a $\treeplus_v$-cut $\multicutM$ the ``boundaries'' $\boundary \multicutM (\pi_v)$ and $\boundary\multicutM$
always refer to the boundary edges restricted to $E(\treeplus_v)$. 
For a non-top component in $\multicutM$, its boundaries in $T$ and $\treeplus_v$ are identical.



\begin{figure}[t]
\begin{center}
\includegraphics[width = 5in]{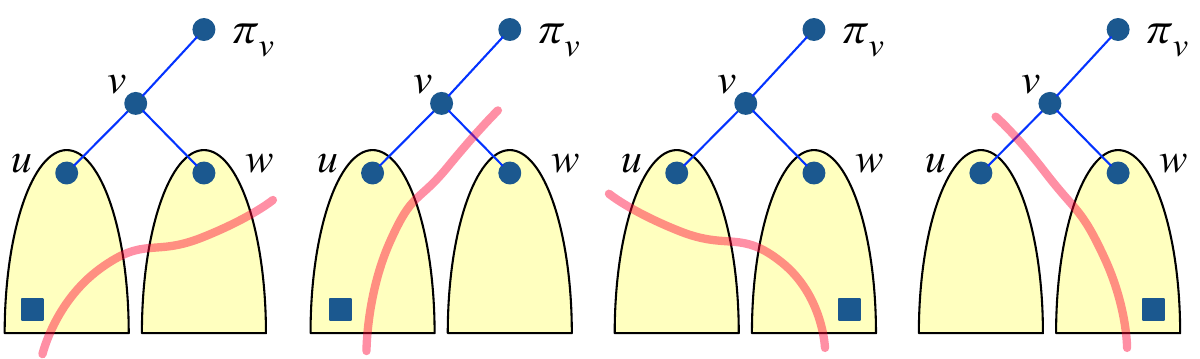}
\caption{%
Illustration of equation~(\ref{eqn: multiway min-max recurrence proper}) for $\dptable_v^1(L)$.
The wide orange (light) curve represents the boundary of the top component of $\dptable_v^1(L)$.
The two first cases correspond to $\sigma = 1$, the two other cases correspond to $\sigma = 0$.}
\label{fig: dp_for_minmax_recurrence}
\end{center}
\end{figure}


\myparagraph{Algorithm for integer weights.} 
For each $v\in V$ we define two functions $\dptable_v^0, \dptable_v^1 \colon \posreals \to \posreals$. 
$\dptable_v^1 (L)$ is the minimum, over all proper $\treeplus_v$-cuts $\multicutM$ with all non-top boundaries of weight at most $L$, 
of the boundary weight  $\weight(\boundary \multicutM(\pi_v) )$ of its top component in $\treeplus_v$. 
(When no $\treeplus_v$-cut satisfies these constraints, then $\dptable_v^i(L)=\infty$.)
$\dptable_v^0 (L)$ is defined analogously, but over improper $\treeplus_v$-cuts.
We colloquially refer to these as the ``D-functions''. 

We first fix $L$, and show a dynamic programming algorithm that computes $\dptable_v^0(L)$ and $\dptable_v^1(L)$
for all $v\in V$.
These values are stored in a table, and computed in the bottom-up order using a recurrence equation.
If $t$ is a leaf (that is, $t \in \Terminals$), we have $D^1_t (L)=0$ (no edges cut), and $\dptable_t^0 (L)$ equals $\weight(t\pi_t)$ if $\weight(t\pi_t) \le L$ 
(cut the edge $t\pi_t$) and is $\infty$ otherwise. 
For an internal vertex $v\in V$ with children $u$ and $w$, we use the following recurrences:
\begin{align}
	\dptable_v^1(L) \;&=\; \min_{\sigma \in \braced{0,1} } \braced{\dptable_u^{\sigma} (L) + \dptable_w^{1-\sigma} (L)} 
	\label{eqn: multiway min-max recurrence proper}
	\\
	\dptable_v^0 (L) \;&=\; \begin{cases}
		\min\braced{\dptable_u^0 (L)+\dptable_w^0(L), \weight(v\pi_v)} & \quad \text{if } \weight(v\pi_v) + \dptable_v^1 (L) \le L \\
		\dptable_u^0 (L)+\dptable_w^0(L) & \quad \text{else}
	\label{eqn: multiway min-max recurrence improper}
	\end{cases}
\end{align}
Figure~\ref{fig: dp_for_minmax_recurrence} illustrates equation~(\ref{eqn: multiway min-max recurrence proper}) for $\dptable_v^1(L)$.
For $\sigma =1$, the top component of $\dptable_v^1(L)$ is obtained from a top (proper) component of $\dptable_u^1(L)$
and the (improper) top component (which may or may not contain $w$) of $\dptable_w^0(L)$. The case of $\sigma=0$
is symmetric. The intuition behind the recurrence for  $\dptable_v^0(L)$ is similar, except that we also need to take into
account the case that its top compoment may only contain $\pi_v$.

From the definitions, we have that $\maxopt(T,\ConflictGraph) \le L$ if and only if $\dptable_r^1 (L) \le L$. 
Thus, checking whether $\maxopt(T,\ConflictGraph) \le L$ can be done with a subroutine that
computes the D-values at all nodes of the tree by the recurrence~(\ref{eqn: multiway min-max recurrence proper})
and (\ref{eqn: multiway min-max recurrence improper}), and
then checks if $\dptable_r^1 (L) \le L$.

If the weights are integer, the overall algorithm can compute the optimal value $L$ by
performing binary search for $L$ in the interval $[0,W]$, in each iteration using the above subroutine to 
test whether $\maxopt(T,\ConflictGraph) \le L$.

To complete the proof of Theorem~\ref{thm: poly algorithm for min-max for trees} for integer weights, we still need to
justify that all formulas for $\dptable_v^0(L)$ and $\dptable_v^1(L)$ are correct. 
The proof is by induction. If $v$ is a leaf, the correctness follows directly from definition. So assume that 
$v$ is an internal vertex. 

To streamline the argument, some auxiliary terminology will be helpful.
Suppose $v$ has children $u$ and $w$.
A $\treeplus_v$-cut is \emph{$L$-feasible} if all its non-top boundaries have weight at most $L$.
If $\multicutQ_u$ and $\multicutQ_w$ are a $\treeplus_u$-cut and $\treeplus_w$-cut respectively with top components $X_u$ and $X_w$,
and one of them is improper, their \emph{merge} is a
$\treeplus_v$-cut $\multicutQ_v$ consisting of all non-top components of $\multicutQ_u$ and $\multicutQ_w$
and top component $X_v = X_u\cup X_w \cup\braced{\pi_v}$.
Clearly, $|X_v \cap \Terminals| \le 1$ 
and $\weight(\partial X_v) = \weight(\partial X_u) + \weight(\partial X_w)$.
(Recall our convention that the boundary of a top component of a subtree $T^+_z$ includes only edges in $T^+_z$.)
Also, $\multicutQ_v$ is $L$-feasible iff both $\multicutQ_u$ and $\multicutQ_w$ are $L$-feasible.
Further, \emph{any} $\treeplus_v$-cut $\multicutQ_v$ (with top component $X_v$ containing $v$) induces
a $\treeplus_u$-cut $\multicutQ_u$ and a $\treeplus_w$-cut $\multicutQ_w$ such that
$\multicutQ_v$ is a merge of $\multicutQ_u$ and $\multicutQ_w$: assign each
non-top component to $\multicutQ_u$ or  $\multicutQ_w$, depending on in which subtree it is contained,
and let the respective top components be $X_u = X_v\cap V(T^+_u)$ and $X_w = X_v\cap V(T^+_w)$.

We are now ready to continue the proof. Consider first $\dptable_v^1(L)$. 
Recall that in the normal form all terminals are leaves. By the above discussion,
any $\treeplus_v$-cut $\multicutQ_v$ is proper and $L$-feasible if and only if it is a
merge of an $L$-feasible $\treeplus_u$-cut $\multicutQ_u$ and an $L$-feasible $\treeplus_w$-cut $\multicutQ_w$,
one proper and one improper.
The correctness of~(\ref{eqn: multiway min-max recurrence proper}) follows from the earlier established
relation between the top boundaries, namely $\weight(\partial X_v) = \weight(\partial X_u) + \weight(\partial X_w)$.

Next, consider $\dptable_v^0(L)$.
We distinguish two disjoint types of improper $L$-feasible $\treeplus_v$-cuts $\multicutQ_v$: 
(i) those whose top component $X_v$ is $X_v = \braced{\pi_v}$ and (ii) those where $X_v \neq \braced{\pi_v}$.
It is sufficient to show that the two options in the formula~(\ref{eqn: multiway min-max recurrence improper}) 
correctly compute the optimal improper $L$-feasible $\treeplus_v$-cuts of these two types, respectively.

Regarding type~(i), let $\multicutQ_v$ be an improper $L$-feasible $\treeplus_v$-cut.
Observe that $\multicutQ_v$ is
of type~(i) iff there is a \emph{proper} $L$-feasible $\treeplus_v$-cut $\multicutQ'_v$ 
such that the non-top components of $\multicutQ_v$ consist of all non-top components of $\multicutQ'_v$
and one additional component $X''_v = X'_v\setminus\braced{\pi_v}$ 
that satisfies $\weight(\cutof X''_v) = \weight(v\pi_v) + \weight(\cutof X'_v) \le L$,
where $X'_v$ is the top component of $\multicutQ'_v$. (We add
$\weight(v\pi_v)$ because the boundary of $X''_v$, treated as a non-top component in $\multicutQ_v$, 
also includes $v\pi_v$. Note that there is no type~(i) $\treeplus_v$-cut if this inequality is false.)
This shows the correctness of the top option in~(\ref{eqn: multiway min-max recurrence improper}).

The argument for type~(ii) is easier:
By the earlier discussion,
$\treeplus_v$-cut $\multicutQ_v$ is of type~(ii) if and only if it is 
a merge of an improper $L$-feasible $\treeplus_u$-cut $\multicutQ_u$ and an improper $L$-feasible $\treeplus_w$-cut $\multicutQ_w$.
Then again the correctness of~(\ref{eqn: multiway min-max recurrence improper}) for type~(ii)
follows from the earlier established
relation between the top boundaries, namely $\weight(\partial X_v) = \weight(\partial X_u) + \weight(\partial X_w)$.


\myparagraph{Extension to real weights.} 
The key idea is not to deal with \emph{individual} values $D_v^\sigma(L)$ of
the functions, but rather to deal with the \emph{entire} functions
$D_v^\sigma(\cdot)$. This is possible because the functions have a very simple
structure.

\begin{claim}\label{cla: D-functions form}
	For each $v\in V$ and $\sigma\in\braced{0,1}$, the function $\dptable_v^\sigma (\cdot)$ is non-increasing, right-continuous, and piece-wise constant.\footnote{
	To account for $\infty$-values, ``right-continuous'' in this context means that for each $0\le L<\infty$, there is some $\epsilon>0$ such that $\dptable_v^i(L+\delta)=\dptable_v^i(L)$ for all $0<\delta<\epsilon$.}
	Furthermore, all points of discontinuity of all D-functions are from some set $\jumps \subseteq \posreals$ of size $n=|V|$.
\end{claim}

\begin{proof}
The fact that $\dptable_v^\sigma$ is non-increasing follows immediately from its definition. To prove the rest of the claim, it is useful to introduce a set of auxiliary functions
$H_v$:
for $v\in V$, let $H_v(L)=\weight(v\pi_v)$ if $\weight(v\pi_v)+\dptable_v^1 (L) \le L$, and otherwise let $H_v(L)=\infty$.
Then for $t\in\Terminals$, $\dptable_t^0 (L) = H_t (L)$,
and for each internal vertex $v \in V \setminus \Terminals$ with children $u,w$,
$\dptable_v^0 (L) = \min\braced{\dptable_u^0 (L) +\dptable_w^0 (L), H_v (L)}$. 
By induction,  all D-functions are some combination of $H$-functions using component-wise $\min$ and $+$ operations.

The condition in the definition of $H_v$ 
can be phrased as an optimization problem: $\weight(v\pi_v)+\dptable_v^1 (L) \le L$ if and only if 
there is an \emph{improper} $\treeplus_v$-cut which cuts $v\pi_v$ and has min-max cost (in $\treeplus_v$) at most $L$.
(Or equivalently, the component containing $\pi_v$ is a singleton. In this interpretation, $\dptable_v^1 (L)$ corresponds to the boundary of the component containing $v$ in such a cut.)
Then $L_v \coloneqq \min \braced{L \suchthat \weight(v\pi_v)+\dptable_v^1 (L) \le L}$ is the minimum max-cost over all such cuts, 
and $H_v(L) = \weight(v\pi_v)$ if $L \ge L_v$, and $H_v(L)=\infty$ otherwise.
That is, $H_v$ is just a right-continuous step function with at most one point of discontinuity.
Then, by induction, the points of discontinuity of $\dptable_v^\sigma$ form a subset of $\jumps_v^\sigma$, 
where $\jumps_v^1 =\braced{L_u}_{u \in T_v -\braced{v}}$ and $\jumps_v^0 = \jumps_v^1 \cup \braced{L_v}$.
It also follows that $\dptable_v^\sigma$ is piece-wise constant and right-continuous.
\end{proof}


In the algorithm, we represent a piece-wise constant right-continuous function $f$
by a list $\braced{(a_0,f(a_0)),...,(a_l,f(a_l))}$ where $a_1 < ... < a_l$ are the points of discontinuity and $a_0=0$. 
For $x\ge 0$, $f(x)=f(a_i)$, where $i$ is the largest index such that $a_i \le x$.
Given list-representations for functions $f$ and $g$ with discontinuity sets $A$ and $B$ respectively, we can compute a list 
representation for their component-wise minimum or sum, whose discontinuities are in $A\cup B$, in $O(|A|+|B|)$ time.
Crucially, because the points of discontinuity among our $D$-functions come from a subset of size $n$, we can perform these calculations in $O(n)$ time.

The final algorithm is this: for each $v\in V$ and $\sigma\in \braced{0,1}$, 
compute the value $L_v$ and a list representation of $\dptable_v^\sigma$ over discontinuities $\jumps_v^\sigma$.
Then return $L_r$, which is equal to $\maxopt(T,\ConflictGraph)$.
Computing lists for the terminals is easy, so suppose $v$ is an internal vertex with children $u,w$.
Computing a list for $\dptable_v^1$ takes $O(n)$ time from functions $\dptable_u^0$, $\dptable_w^0$, $\dptable_u^1$ and $\dptable_w^1$,
using equation~(\ref{eqn: multiway min-max recurrence proper}).
To compute a list for $\dptable_v^0$, we first need to compute $L_v$ from $\dptable_v^1$.
To do this, suppose $\dptable_v^1$ has points of discontinuity $a_1,...,a_l$ and $a_0=0$.
We then claim $L_v=\min_{0\le i \le l} a'_i$, where $a'_i= \max\braced{a_i, \weight(v\pi_v) + \dptable_v^1 (a_i)}$.
Indeed, if $j$ is the greatest index where $a_j \le L_v$, then $\weight(v\pi_v) + \dptable_v^1 (a_j)\le L_v$ and $a'_j \le L_v$.
And by monotonicity, $\dptable_v^1(a'_i)+\weight(v\pi_v) \le a'_i$ for all $i$, so the formula is correct.
Then using the equation $\dptable_v^0 (L) = \min\braced{\dptable_u^0 (L) +\dptable_w^0 (L), H_v (L)}$ we compute a list for $\dptable_v^0$.
Clearly, we spend $O(n)$ computing a list for each $\dptable_v^\sigma$, and thus the total running time is $O(n^2)$,
completing the proof of Theorem~\ref{thm: poly algorithm for min-max for trees}.


\section{Algorithms Parameterized by Neighborhood Diversity}
\label{sec: neighborhood diversity algorithms}

	In this section we prove Theorem~\ref{thm: neighborhood diversity result}(a). 
	We assume that our tree $T$ is in binary normal form with root $r$.
	We add a dummy parent $\pi_r$ for $r$ and set $\weight(r \pi_r)=0$.
	We will also use the notion of $\treeplus_v$-cuts and their boundary notation as introduced in Section~\ref{sec: polynomial time algorithm multiway minmax}.
	
	To motivate our algorithm, we briefly sketch a simple (albeit inefficient) dynamic program for $\MulticutSumTrees$.
	Its pseudo-code is presented in Algorithm~\ref{alg: ind set warm up}.
	$\dptable_v^s$ and $\dptable_v^*$ have a natural interpretation: for an independent set $s\subseteq \Terminals$,
	$\dptable_v^s$ is the minimum sum-cost over all $\treeplus_v$-cuts where the top cluster contains exactly the terminals from $s$, while $\dptable_v^*$ is the minimum sum-cost over all proper $\treeplus_v$-cuts.
	
	Without giving a complete analysis, we quickly justify the recurrences in Lines~\ref{line: min sum convolution warm up non empty}~and~\ref{line: min sum convolution warm up empty}.
	For vertex $v\in V - \Terminals$ with children $u$ and $w$, let $\multicutM$ be a $\treeplus_v$-cut whose top component satisfies $\multicutM(\pi_v) \cap \Terminals = s$.
	If $v\pi_v \notin \boundary\multicutM$ (which holds when $s\ne \emptyset$), then 
	restricting the components of $\multicutM$ to $\treeplus_u$ and $\treeplus_w$ yields a $\treeplus_u$-cut $\multicutM_u$ and a $\treeplus_w$-cut $\multicutM_w$, respectively.
	The terminals in the corresponding top components $\multicutM_u(\pi_u)$ and $\multicutM_w(\pi_w)$ partition $s$ (perhaps trivially, one or both components may have no terminals),
	while the boundary edges $\boundary \multicutM = \boundary \multicutM_u \sqcup \boundary \multicutM_w$ get partitioned.
	This justifies Line~\ref{line: min sum convolution warm up non empty} and the first term in Line~\ref{line: min sum convolution warm up empty}.
	When $v\pi_v \in \boundary\multicutM$, $\multicutM$ is improper and removing $v\pi_v$ from the boundary creates a proper $\treeplus_v$-cut,
	justifying the second term in Line~\ref{line: min sum convolution warm up empty}.


\begin{algorithm}[t]
\caption{Exponential-time algorithm for computing $\sumopt(T,\ConflictGraph)$.}\label{alg: ind set warm up}
\begin{algorithmic}[1]
\State Let $\states$ be the family of all independent sets of $\ConflictGraph$, including $\emptyset$
\State For each $v\in V$ and $s\in\states$, create a table entry $\dptable_v^s \in\posreals\cup\braced{\infty}$
\For{$t \in \Terminals$}
	\Comment{initialize $\dptable_t$ for all leaves $t\in \Terminals$}
	\State $\dptable_t^{\braced{t}} \gets 0$ \,,\, 
			$\dptable_t^* \gets 0$ \,,\,
			$\dptable_t^{\emptyset} \gets \weight (t\pi_t)$ \,,\, 
			$\dptable_t^s \gets \infty$ for $s\in \states \setminus \braced{\braced{t}, \emptyset}$\label{line: warm up base case}
\EndFor
\For{$v\in V-\Terminals$, processing bottom-up}
	\State Let $u$ and $w$ be $v$'s children
	\For{$s \in \states - \braced{\emptyset}$}
		\State $\dptable_v^s \gets \min_{a \add b=s} \braced{\dptable_u^a+\dptable_w^b}$ \label{line: min sum convolution warm up non empty}
		\Comment{$a \add b=s$ is a partition of $s$}
	\EndFor
	\State $D^*_v \gets \min_{s \ne \emptyset} \dptable_v^s$
	\State $\dptable_v^\emptyset \gets \min\braced{\,\dptable_u^\emptyset+\dptable_w^\emptyset\,,\, \weight(v\pi_v) + D^*_v\,}$. \label{line: min sum convolution warm up empty}
\EndFor
\State \Return $D^*_r$
\end{algorithmic}
\end{algorithm}


	The runtime is dominated by the number of independent sets in $\ConflictGraph$ and Line~\ref{line: min sum convolution warm up non empty}, 
	so it's at most $O(|\states|^2 n)$ in the worst case.\footnote{
	In terms of $|\Terminals|$, one can see that the running time is in fact $O(3^{|\Terminals|} n)$, since Line~\ref{line: min sum convolution warm up non empty} is just a min-sum subset convolution.
	Note also that  independent sets can be enumerated in $O(|\states|n)$ time \cite{kurita_etal_constant_amortized_time_enumeration_for_independent_sets_in_graphs_with_bounded_clique_number_2021}.}
	
	A natural idea for speeding up the algorithm is to group independent sets into equivalence classes, as few as possible, that can be used as states in dynamic programming. 
	Roughly, two independent subsets $s,s'$ of
	$\Terminals$ can be considered equivalent if, for any vertex $v$, swapping $s$ and $s'$ in the top cluster of $v$
	will not affect the computation in the ancestors of $v$.
	To generalize the recurrence in line~\ref{line: min sum convolution warm up non empty} of the algorithm, this equivalence relation
	needs to be preserved under disjoint unions. 
	For our algorithm, we define this equivalence relation on independent sets in terms of a \emph{type graph} and
	\emph{neighborhood diversity}~\cite{lampis_algorithmic_meta_theorems_for_restrictions_of_treewidth_2012}, 
	which we define next.
	This relation is in general not optimal, but it applies to a broad class of graphs and can be computed efficiently.
	Later, in Section~\ref{sec: generalized type graphs} we consider a natural generalization of these concepts
	that allows us to establish Corollary~\ref{cor: neighborhood diversity + clique cover + multipartite result}.

	
	\myparagraph{Neighborhood diversity.}
	For a conflict graph $\ConflictGraph=(\Terminals, \ConflictEdges)$, two terminals $t,t'\in\Terminals$ are \emph{of the same type} (also called \emph{twins}) if $N(t')-\braced{t}=N(t)-\braced{t'}$,
	where $N(\cdot)$ is the open neighborhood in $\ConflictGraph$. This property forms an equivalence relation, whose congruence classes $\TwinClasses$ are called \emph{types}.
	The number of types $\nd(\ConflictGraph)=|\TwinClasses|$ is then the \emph{neighborhood diversity} of $\ConflictGraph$.
	We treat types as abstract elements instead of subsets of terminals, using the map $\typeof\colon \Terminals \to \TwinClasses$ to denote which type $\typeof(t)$ that terminal $t$ is assigned to.
	Associated with the types is a unique \emph{type graph} $\TypeGraph=(\TwinClasses, \TwinEdges, \typeof)$, which has the property that for all $t,t'\in \Terminals$,
	$tt' \in \ConflictEdges$ if and only if $\typeof(t)\typeof(t')\in \TwinEdges$.
	(A type $\eta\in\TwinClasses$ may have a self-loop in $\TypeGraph$, which we denote by $\eta \eta \in \TwinEdges$.)
	Intuitively, the types form a partition of $\Terminals$ where (i) each part is either a clique or independent set and
	(ii) between each pair of parts is either a biclique or no edges at all.
	It is a standard fact that $\TypeGraph$ can be computed in $O(n+|\ConflictEdges|)$ time (see Figure~\ref{fig: type graph} for an example).
	

	\begin{figure}[t]
	\begin{center}
	\includegraphics[width = 4.5in]{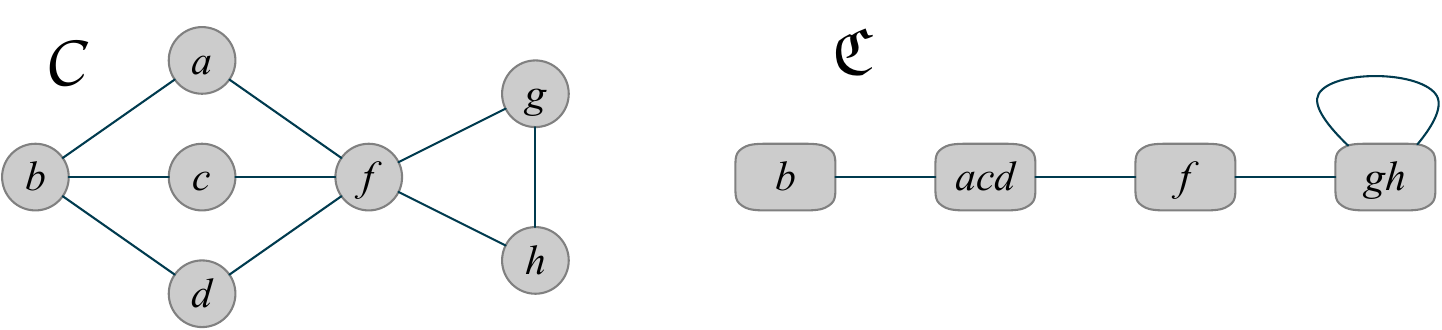}
	\caption{%
	An example of a conflict graph $\ConflictGraph$ and its type graph $\TypeGraph$.
	The types are labeled by the vertices of $\ConflictGraph$ mapped to them.
	Since $\TypeGraph$ has four vertices (types), $\nd(\ConflictGraph) = 4$.
	}
	\label{fig: type graph}
	\end{center}
	\end{figure}



	\myparagraph{States and splits.}
	Let $\CTwinClasses\subseteq\TwinClasses$ denote the types with self-loops (\emph{c-types}) and $\ITwinClasses=\TwinClasses-\CTwinClasses$ (\emph{i-types}).
	Note that the terminal set $\typeof^{-1}(\eta)$ is a clique when $\eta$ is a c-type and is an independent set otherwise.
	Letting $\typeof(\Terminals') = \braced{\typeof(t) \mid t \in \Terminals'}$ be the image of $\Terminals'\subseteq\Terminals$ under $\typeof$,
	we observe that $\typeof$ maps independent sets of $\ConflictGraph$ to independent sets of $\TypeGraph$ (with self-loops allowed).
	From here on we refer to the latter as \emph{states} and denote its elements $s\in\indsets(\TypeGraph)$ using lower-case letters. (We include $\emptyset \in \indsets(\TypeGraph)$.)
	Naturally, $\indsets(\TypeGraph)$ will replace $\indsets(\ConflictGraph)$ in Algorithm~\ref{alg: ind set warm up},
	but to do so requires obtaining an analog of Line~\ref{line: min sum convolution warm up non empty};
	here we need a correspondence between bipartitions $\Terminals'=A\sqcup B$ of an independent set 
	$\Terminals'\in \indsets(\ConflictGraph)$ and pairs of states associated with $\typeof(\Terminals')$.
	Crucially, we in general do \emph{not} want to consider bipartitions of $\typeof(\Terminals')$:
	while the c-types in $\typeof(\Terminals')$ will be partitioned between $\typeof(A)$ and $\typeof(B)$, if $\eta\in\typeof(\Terminals')$ is an i-type,
	then it's possible that $A$ and $B$ could \emph{both} contain terminals from $\typeof^{-1}(\eta)$.
	
	This motivates the following definitions.
	For $s\in\indsets(\TypeGraph)$, $\ctwin{s}$ and $\itwin{s}$ are the c-types and i-types in $s$ respectively, so $s=\ctwin{s}\sqcup\itwin{s}$.
	The \emph{splits of $s$} denotes the collection of pairs $\splits(s) = \braced{(a,b) \in \subsets(\TwinClasses)^2 \mid a\cup b = s \text{ and } \ctwin{a} \cap \ctwin{b} = \emptyset}$.
	In particular, we allow $\itwin{a} \cap \itwin{b}  \ne \emptyset$.
	
	Our algorithm does not depend so much on the particulars of neighborhood diversity or type graphs.
	In fact, we only need the three following properties:
	\begin{description}
	\item[(A1)] For all $A \in \indsets(\ConflictGraph)$, $A =\emptyset$ if and only if $\typeof(A)=\emptyset$.
	\item[(A2)] For all $B \in \indsets(\ConflictGraph)$ and $A\subseteq B$, $(\typeof(A), \typeof(B-A)) \in \splits(\typeof(B))$.
	\item[(A3)] For all $A,B \in \indsets(\ConflictGraph)$ where $A \cap B = \emptyset$, $(\typeof(A), \typeof(B)) \in \splits(s)$ implies $A\cup B \in \indsets(\ConflictGraph)$ and $\typeof(A \cup B)=s$.
	\end{description}
	Intuitively, (A1) just says that $\typeof$ distinguishes the empty set in $\indsets(\ConflictGraph)$, 
	(A2) says we can convert bipartitions of $B \in\indsets(\ConflictGraph)$ into splits of $\typeof(B)$,
	(A3) says that $A\sqcup B$ is automatically independent if $(\typeof(A),\typeof(B))$ is a split of some state (and that $\typeof(A\cup B)$ must be this state).
	That these properties hold is straightforward to verify.
	
	The runtime of our algorithm is dominated by the total number of splits or \emph{split complexity}, denoted by $\spcomp(\TypeGraph)=\sum_{s \in \indsets(\TypeGraph) } |\splits(s)|$.
	It is easy to see that $\splits(s)$ and $\splits(s')$ are disjoint for two distinct states $s$, $s'$,
	so the total number of splits is $\spcomp(\TypeGraph)\le |\indsets(\TypeGraph)|^2 \le 4^{\nd(\ConflictGraph)}$.


\myparagraph{Algorithm for $\MulticutSumTrees$.}
	As mentioned before, we can compute $\ConflictGraph$'s type graph $\TypeGraph$ in $O(n+|\ConflictEdges|)$ time.
	The main algorithm will actually run in time $O(\spcomp(\TypeGraph) n)$, which we do by enumerating $\indsets(\TypeGraph)$ and the splits directly.
	Enumerating independent sets can be done in $O(|\indsets(\TypeGraph)| |\TwinClasses|)=O(\spcomp(\TypeGraph) n)$ time, and for each $s\in\indsets(\TypeGraph)$,
	$\splits(s)$ can be computed in $O(|\splits(s)|)$ time by listing all bipartitions of $\ctwin{s}$ and pairs of subsets of $\itwin{s}$.
	
	We now just need to modify Algorithm~\ref{alg: ind set warm up} as follows: $\states$ is now $\indsets(\TypeGraph)$, 
	every instance of $\braced{t}$ in Line~\ref{line: warm up base case} is replaced by $\braced{\typeof(t)}$,
	and the min-sum convolution formula in Line~\ref{line: min sum convolution warm up non empty} is replaced by $\min_{(a,b)\in\splits(s)}\braced{D_u^a+D_w^b}$.
	We also reinterpret the meaning of $D_v^s$ and $D_v^*$, as follows. 
	We say a $\treeplus_v$-cut $\multicutM$ is in state $s\in\states$ if the top component satisfies $\typeof(\multicutM(\pi_v)\cap \Terminals)=s$.
	Then $D_v^s$ is the minimum sum-cost over all $\treeplus_v$-cuts in state $s$, and by property~(A1) a $T_v$-cut is in state $\emptyset$ if and only if it is improper,
	so $D_v^*$ is still the minimum sum-cost over all proper $\treeplus_v$-cuts.
	
	The correctness of $D_t^s$ for leaves $t \in\Terminals$ in Line~4 is easy to verify.
	The formula in Line~\ref{line: min sum convolution warm up empty} has not changed, so we only need to justify the new formula for Line~\ref{line: min sum convolution warm up non empty}.
	Let $\multicutM$ be a $\treeplus_v$-cut in state $s\ne\emptyset$ with sum-cost $D_v^s$, and for $c\in ch(v)=\braced{u,w}$, 
	let $\multicutM_c$ be the induced $\treeplus_c$-cut by restricting $\multicutM$ to $\treeplus_c$.
	Since $\multicutM$ is proper, $v\pi_v\notin\boundary \multicutM$ and $\multicutM_c(\pi_c)= \multicutM_(\pi_v) \cap \treeplus_c$.
	Letting $\Terminals_c = \multicutM_c(\pi_c) \cap \Terminals$, we clearly have $\boundary\multicutM=\boundary\multicutM_u \sqcup \boundary\multicutM_w$ 
	and $\multicutM(\pi) \cap \Terminals = \Terminals_u \sqcup \Terminals_w$.
	Furthermore, $\multicutM_u$ and $\multicutM_w$ are in states $a = \typeof(\Terminals_u)$ and $b=\typeof(\Terminals_w)$ respectively, 
	and by property~(A2) we must have $(a,b)\in\splits(s)$, and so we have $D_v^s\ge D_u^a+D_w^b$.
	Likewise, if we are given a $\treeplus_u$-cut and $\treeplus_w$-cut with terminals $\Terminals_u$ and $\Terminals_w$ in their respective top-components, 
	and $(\typeof(\Terminals_u), \typeof(\Terminals_w))\in \splits(s)$,
	then by property~(A3) $\Terminals_u \cup \Terminals_w$ is independent and $\typeof(\Terminals_u \cup \Terminals_w)=s$, 
	and so merging the two cuts yields a $\treeplus_v$-cut in state $s$.
	This completes the proof of Theorem~\ref{thm: neighborhood diversity result}(a).


\myparagraph{Algorithm for $\ConMulticutMaxTrees$.}
	The dynamic programming algorithm claimed in Theorem~\ref{thm: neighborhood diversity result}(b)
	is a straighforward combination of the dynamic program formulation above with techniques in 
	Section~\ref{sec: polynomial time algorithm multiway minmax}.
	The key change is that now we need to store the minimum boundary weight of the top component of $\treeplus_v$-cuts in state $s$, 
	subject to the other components having boundary weight at most $L$.
	More details can be found in the proof of a more general statement in
	in Section~\ref{sec: generalized type graphs}.


\section{Generalized Type Graphs and Split Complexity}
\label{sec: generalized type graphs}

	In this section, we extend the algorithms from Sections~\ref{sec: polynomial time algorithm multiway minmax}~and~\ref{sec: neighborhood diversity algorithms}
	using ``generalized'' type graphs, then prove Theorem~\ref{thm: generalized neighborhood diversity result} and Corollary~\ref{cor: neighborhood diversity + clique cover + multipartite result}.

	
\myitparagraph{Basic ideas.}
We start by briefly explaining the concept of generalized type graphs and split complexity, and
outlining the proof of Theorem~\ref{thm: generalized neighborhood diversity result}.
We are assuming the reader is familiar with the terminology from Section~\ref{sec: neighborhood diversity algorithms}.
Informally, a ``generalized type graph'' $\TypeGraph=(\TwinClasses,\TwinEdges,\typeof)$ modifies the notion of type graphs from neighborhood diversity, allowing 
a terminal $t$ to be assigned \emph{multiple} types $\typeof(t) \subseteq\TwinClasses$ rather than one. 
A conflict $tt'\in\ConflictEdges$ is represented by the existence of an edge in $\TwinEdges$ between some pair
of types from $\typeof(t)$ and $\typeof(t')$. This extension can exponentially reduce the number of dynamic programming states: in particular we only need $2^q$ states in this model for $q$-clique cover.
Aside from changing the semantics of types, the definition of splits is unchanged, and in fact $\spcomp(\TypeGraph)$ was already defined in Section~\ref{sec: neighborhood diversity algorithms} for the standard type graph
and is unchanged for our general formulation.
Crucially, the properties~(A1)-(A3) identified in that section still hold in this model,
and the algorithms for Theorem~\ref{thm: generalized neighborhood diversity result} are almost identical to their neighborhood diversity variants.
	
	In general, conflict graphs admitting a type graph $\TypeGraph$ with low split complexity $\spcomp(\TypeGraph)$ must have bounded neighborhood diversity $\nd(\ConflictGraph$),
	but the former can be exponentially smaller than the $4^{\nd(\ConflictGraph)}$ factor in Theorem~\ref{thm: neighborhood diversity result}.
	The downside is that the following algorithms require specifying $\TypeGraph$ on input; because this model extends graphs with $q$-clique covers, which take $O^*(2^{2^q})$ time to detect,
	it is unlikely that efficiently computing a minimal generalized type graph is possible.
	
	%

	\begin{figure}[t]
	\begin{center}
	\includegraphics[width = 6in]{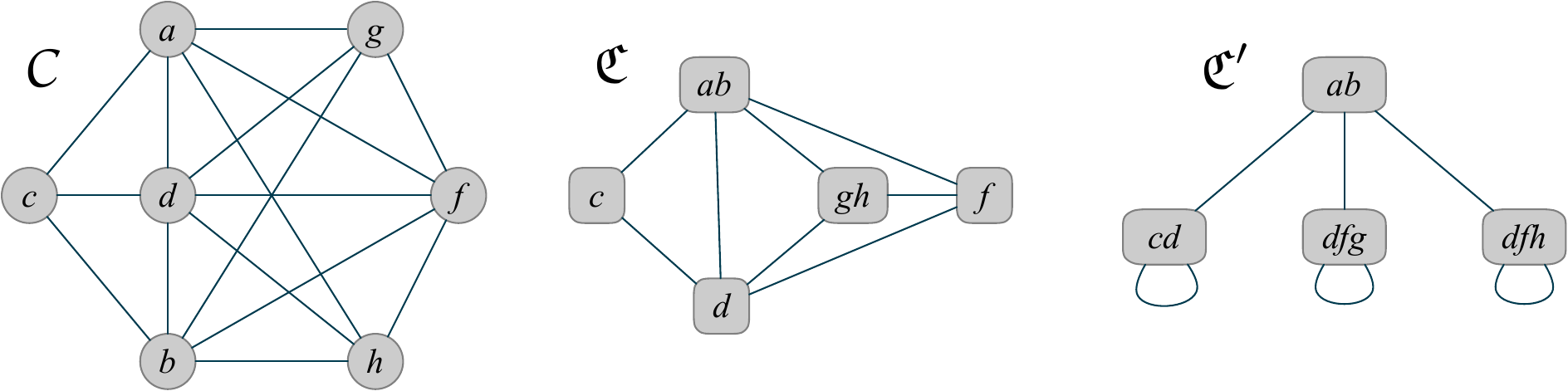}
	\caption{%
	An example of a conflict graph $\ConflictGraph$, its (standard) type graph $\TypeGraph$,
	and one generalized type graph $\TypeGraph'$ for $\ConflictGraph$.
	The types are labeled by the vertices of $\ConflictGraph$ mapped to them.
	In $\TypeGraph'$, vertices may be mapped to multiple types.
	The example shows that in a generalized type graph the number of types can be smaller.
	}
	\label{fig: generalized type graph}
	\end{center}
	\end{figure}
	%


	\myparagraph{Generalized Type Graph.}
	For a conflict graph $\ConflictGraph$, a \emph{generalized type graph} 
	(or simply a type graph) $\TypeGraph = (\TwinClasses,\TwinEdges,\typeof)$
	is an undirected graph (with self-loops allowed) on \emph{types} $\TwinClasses$ connected by edges $\TwinEdges$,
	together with a map $\typeof\colon \Terminals \to \indsets(\TypeGraph)-\braced{\emptyset}$, where $\indsets(\TypeGraph)$ denotes the family of 
	independent sets in $\TypeGraph$ (allowing self-loops).
	In order to be a valid type graph for $\ConflictGraph$, $\TypeGraph$ must satisfy the following condition:
	$tt' \in\ConflictEdges$ if and only if there exists some edge $\eta \eta' \in \TwinEdges$ 
	with $\eta \in \typeof(t)$ and $\eta' \in \typeof(t')$.
	(If $\eta = \eta'$ then $\eta \eta'$ is a self-loop.)
	As in Section~\ref{sec: neighborhood diversity algorithms}, the types $\CTwinClasses\subseteq\TwinClasses$ with self-loops are called c-types
	 and the remaining types $\ITwinClasses=\TwinClasses-\CTwinClasses$ are called i-types.
	For simplicity we also impose $|\TwinClasses|\le|\Terminals|$.\footnote{
	If we were only interested in minimizing the number of types, then trivially we can make this assumption without loss of generality.
	It's not a priori clear that we can make this assumption for minimizing splits, but we'll put this aside.
	}
	
	If $\typeof$ maps terminals to singletons, this is equivalent to the standard type graph used to define neighborhood diversity.
	The crucial difference is that terminals can now be assigned \emph{multiple} types.
	Intuitively, the conflict graphs induced by a generalized type graph with $k$ types can be constructed in the following way:
	select $k$ (possibly overlapping) subsets $\Terminals_1 ,..., \Terminals_k$ of $\Terminals$, 
	make some $\Terminals_i$ cliques, and add a biclique between any \emph{disjoint} pair $\Terminals_i,\Terminals_j$.
	See Figure~\ref{fig: generalized type graph} for an example.
	
	We make a few remarks before moving on. For one, conflict graphs may be represented by distinct generalized type graphs, even if they use a minimal number of types. 
	Second, the terminals $\braced{t \in \Terminals \mid \eta \in \typeof(t)}$ associated with an i-type $\eta$ may not form an independent set, as pairs of some terminals
	may be assigned the same c-type or two types adjacent in $\TypeGraph$.
	Third, we observe that $\typeof(t)=\typeof(t')$ implies that $t$ and $t'$ are twins, which yields the inequality $\nd(\ConflictGraph)\le 2^{|\TwinClasses|}$; in particular,
	a type graph with a constant number of types induces a conflict graph with constant neighborhood diversity.
	

	\myparagraph{States and splits.}
	We now discuss states and splits for a conflict graph $\ConflictGraph$ induced by some generalized type graph $\TypeGraph$.
	Letting $\typeof(\Terminals')=\bigcup_{t\in\Terminals'} \typeof(t) \subseteq \TwinClasses$ for $\Terminals'\subseteq\Terminals$,
	we observe that, just like in Section~\ref{sec: neighborhood diversity algorithms}, $\typeof$ maps independent sets in $\ConflictGraph$ to independent sets in $\TypeGraph$ (perhaps with self-loops).
	Like before, \emph{states} $s\in\indsets(\TypeGraph)$ refer to these independent sets, including the empty set.
	
	The definition of splits is the same as in Section~\ref{sec: neighborhood diversity algorithms}.
	For $s\in\indsets(\TypeGraph)$ we define the splits $\splits(s) = \braced{(a,b) \in \indsets(\TypeGraph)^2 \mid a\cup b = s \text{ and } \ctwin{a} \cap \ctwin{b} = \emptyset}$,
	where $\ctwin{s}$ refers to c-types which $s$ contains.
	
	We illustrate the definitions with two examples.
	If $\ConflictGraph_1$ is a union of $q$ cliques, then it is induced by a type graph $\TypeGraph_1=(\TwinClasses_1, \TwinEdges_1, \typeof_1)$, where $\TwinClasses_1=[q]$, 
	$\TwinEdges_1$ is the set of all self-loops, and $\typeof_1(t)$ indicates which cliques $t$ belongs to.
	We have $\indsets(\TypeGraph_1)=\subsets([q])$ (every subset of types is independent)
	and $\splits_1(s)=\braced{(a,s-a) \mid a \subseteq s}$ (every bipartition is a split of $s$).
	
	Letting $\ConflictGraph_2$ denote a complete $c$-partite graph (that is not a clique), the standard type graph $\TypeGraph_2=(\TwinClasses_2,\TwinEdges_2, \typeof_2)$
	using neighborhood diversity is a complete graph on $c$ types (say $\TwinClasses_2=[c]$) with no self-loops, 
	where $\braced{\typeof^{-1}_2(\eta)}_{\eta\in [c]}$ is simply the $c$ different parts.
	The states $\indsets(\TypeGraph_2)=\braced{\emptyset} \cup \braced{\braced{\eta}\mid \eta \in [c]}$ are exactly the empty set and singleton sets,
	and the splits of $\braced{\eta}$ are $\splits_2(\braced{\eta})=\braced{(\emptyset, \braced{\eta}), (\braced{\eta}, \emptyset), (\braced{\eta},\braced{\eta})}$; 
	the last split corresponds to dividing a non-empty subset of a part into two non-empty subsets.
	(If $\ConflictGraph_2$ is a clique, then the standard type graph is a single node with a self-loop.)
	
	By inspection, the splits and the map $\typeof$ still satisfy properties (A1)-(A3) from Section~\ref{sec: neighborhood diversity algorithms}.

	The runtime will be dominated by the total number of splits, denoted by $\spcomp(\TypeGraph)=\sum_{s\in\indsets(\TypeGraph)} |\splits(s)|$ like in Section~\ref{sec: neighborhood diversity algorithms}.
	The same reasoning as in Section~\ref{sec: neighborhood diversity algorithms} implies $\spcomp(\TypeGraph) \le |\indsets(\TypeGraph)|^2 \le 4^{|\TwinClasses|}$. 
	For the two examples discussed earlier,
	$|\indsets(\TypeGraph_1)|=2^q$ and $\spcomp(\TypeGraph_1) = \sum_{X \subseteq [q]} 2^{|X|} = 3^q$, while
	$|\indsets(\TypeGraph_2)|=c+1$ and $\spcomp(\TypeGraph_2)=3c+1$.


	\myparagraph{Algorithm for $\MulticutSumTrees$.}
	We now prove Theorem~\ref{thm: generalized neighborhood diversity result} for $\MulticutSumTrees$.
	Using the type graphs $\TypeGraph_1$ and $\TypeGraph_2$ from earlier, Corollary~\ref{cor: neighborhood diversity + clique cover + multipartite result} follows for the min-sum case.
	
	Given $\TypeGraph=(\TwinClasses,\TwinEdges,\typeof)$, we can enumerate $\indsets(\TypeGraph)$ in time $O(|\indsets(\TypeGraph)||\TwinClasses|  ) = O(|\indsets(\TypeGraph)|n)$
	\cite{kurita_etal_constant_amortized_time_enumeration_for_independent_sets_in_graphs_with_bounded_clique_number_2021}. 
	For each state $s\in\indsets(\TypeGraph)$ we may compute $\splits(s)$ in $O(|\splits(s)|)$ time by enumerating all partitions of $\ctwin{s}$ and all pairs of subsets of $\itwin{s}$.
	
	We now just need to modify Algorithm~\ref{alg: ind set warm up} as follows: $\states$ is now $\indsets(\TypeGraph)$, 
	every instance of $\braced{t}$ in Line~\ref{line: warm up base case} is replaced by $\typeof(t)$,
	and the min-sum convolution formula in Line~\ref{line: min sum convolution warm up non empty} is replaced by $\min_{(a,b)\in\splits(s)}\braced{D_u^a+D_w^b}$.
	We now just need to reinterpret the meaning of $D_v^s$ and $D_v^*$. 
	We say a $\treeplus_v$-cut $\multicutM$ is in state $s\in\states$ if the top component satisfies $\typeof(\multicutM(\pi_v)\cap \Terminals)=s$.
	Then $D_v^s$ is the minimum sum-cost over all $\treeplus_v$-cuts in state $s$, and by property~(A1) a $\treeplus_v$-cut is in state $\emptyset$ if and only if it is improper,
	so $D_v^*$ is still the minimum sum-cost over all proper $\treeplus_v$-cuts.
	
	The correctness of $D_t^s$ for leaves $t \in\Terminals$ in Line~4 is easy to verify.
	The formula in Line~\ref{line: min sum convolution warm up empty} has not changed, so we only need to justify the new formula for Line~\ref{line: min sum convolution warm up non empty}.
	Let $\multicutM$ be a $\treeplus_v$-cut in state $s\ne\emptyset$ with sum-cost $D_v^s$, and for $c\in ch(v)=\braced{u,w}$, 
	let $\multicutM_c$ be the induced $\treeplus_c$-cut by restricting $\multicutM$ to $\treeplus_c$.
	Since $\multicutM$ is proper, $v\pi_v\notin\boundary \multicutM$ and $\multicutM_c(\pi_c)= \multicutM_(\pi_v) \cap \treeplus_c$.
	Letting $\Terminals_c = \multicutM_c(\pi_c) \cap \Terminals$, we clearly have $\boundary\multicutM=\boundary\multicutM_u \sqcup \boundary\multicutM_w$ 
	and $\multicutM(\pi_v) \cap \Terminals = \Terminals_u \sqcup \Terminals_w$.
	Furthermore, $\multicutM_u$ and $\multicutM_w$ are in states $a = \typeof(\Terminals_u)$ and $b=\typeof(\Terminals_w)$ respectively, 
	and by property~(A2) we must have $(a,b)\in\splits(s)$, and so we have $D_v^s\ge D_u^a+D_w^b$.
	Likewise, if we were given a $\treeplus_u$-cut and $\treeplus_w$-cut with terminals $\Terminals_u$ and $\Terminals_w$ in their respective top-components, 
	and $(\typeof(\Terminals_u), \typeof(\Terminals_w))\in \splits(s)$,
	then by property~(A3) $\Terminals_u \cup \Terminals_w$ is independent and $\typeof(\Terminals_u \cup \Terminals_w)=s$, 
	and so merging the two cuts yields a $\treeplus_v$-cut in state $s$.

	
	\myparagraph{Algorithm for $\ConMulticutMaxTrees$.}
	We now prove Theorem~\ref{thm: generalized neighborhood diversity result} for $\ConMulticutMaxTrees$.
	Using the type graphs $\TypeGraph_1$ and $\TypeGraph_2$, Corollary~\ref{cor: neighborhood diversity + clique cover + multipartite result} for $\ConMulticutMaxTrees$ follows.
	We assume that our tree $T$ is in binary normal form with root $r$, and we add a dummy parent $\pi_r$ for $r$ with $\weight(r \pi_r)=0$.
	We use the same notion of $\treeplus_v$-cuts and their boundary convention as in the previous sections. 
	
	Let $\TypeGraph=(\TwinClasses,\TwinEdges,\typeof)$ be a generalized type graph for $\ConflictGraph$ given on input and let $\states=\indsets(\TypeGraph)$.
	As explained before, in $O(\spcomp(\TypeGraph)n)$ time we can enumerate both $\states$ and the corresponding splits.
	
	\paragraph*{Algorithm For Integer Weights.}
	We now describe the main dynamic sub-routine, which in time $O(n \sum_{s\in \states} |\splits(s)|)$ determines whether $\maxopt(T,\ConflictGraph)\le L$ for a given real $L\in\posreals$.
	This automatically implies Theorem~\ref{thm: generalized neighborhood diversity result}(b) for integer weights.
	For now we treat $L$ as fixed.
	
	As in Section~\ref{sec: neighborhood diversity algorithms}, we say a $\treeplus_v$-cut $\multicutM$ is in state $s\in\states$ if the top component satisfies $\typeof(\multicutM(\pi_v) \cap \Terminals)=s$.
	Furthermore, $\multicutM$ is $L$-feasible if all non-top boundaries have weight at most $L$.
	For each $s\in\states$ and $v\in V$, we define the table entry $\dptable_v^s (L)$ to be the minimum, over all $L$-feasible $\treeplus_v$-cuts $\multicutM$ in state $s$,
	of the boundary weight $\weight(\boundary \multicutM(\pi_v))$ of its top component in $\treeplus_v$. (When no such cut exists, set $\dptable_v^s(L)=\infty$.)
	We also create a table entry $\dptable_v^* (L)=\min_{s\ne \emptyset} \dptable_v^s (L)$, denoting the minimum top-boundary weight over all $L$-feasible proper $\treeplus_v$-cuts.
	(By Property~(A1), $\multicutM$ is in state $\emptyset$ if and only if it is improper.)
	From the definitions, $\maxopt(T,\ConflictGraph)\le L$ if and only if $\dptable_r^*(L) \le L$.
	
	For a leaf $t\in \Terminals$, we set $\dptable_t^{\typeof(t)} (L)=0$,
	set $\dptable_t^{\emptyset} (L)=\weight(t\pi_t)$ if $\weight(t\pi_t) \le L$ and $\infty$ otherwise, and set $\dptable_t^s(L)=\infty$ for $s\notin\braced{\typeof(t),\emptyset}$.
	We now describe the recurrence for an internal vertex $v\notin \Terminals$ with children $u,w$.
	\begin{align}
		\dptable_v^{s} (L) &= \min_{(a,b) \in \splits(s)} \braced{\dptable_u^a (L) + \dptable_w^b (L)} \quad\qquad \text{for } s \ne \emptyset \label{eq: nd minmax nonempty state} \\ 
		\dptable_v^* (L) &= \min_{s \ne \emptyset} \dptable_v^s (L)   \label{eq: nd minmax proper cuts}\\ 
		\dptable_v^{\emptyset}(L) &=
		\begin{cases}
			\min\braced{\dptable_u^{\emptyset}(L)+\dptable_w^{\emptyset} (L), \weight(v\pi_v) }	&\quad \text{if } \weight(v\pi_v) + D_v^* (L) \le L \\
			\dptable_u^{\emptyset}(L)+\dptable_w^{\emptyset} (L)							&\quad \text{else}
		\end{cases} \label{eq: nd minmax empty state}
	\end{align}
	The runtime is dominated by the evaluation of Eq.~(\ref{eq: nd minmax nonempty state}), which takes $O(\sum_{s\in \states} |\splits(s)|)$ time per vertex.
	
	The formula for the leaves holds from the definitions and inspection.
	For an internal vertex $v \notin \Terminals$, observe that $\dptable_v^\emptyset (L)$ and $\dptable_v^* (L)$
	have the same semantics as the table entries $\dptable_v^0 (L)$ and $\dptable_v^1 (L)$ from Section~\ref{sec: polynomial time algorithm multiway minmax};
	in particular $\dptable_v^\emptyset$ optimizes over improper cuts, while $\dptable_v^*$ optimizes over proper cuts.
	Indeed, Eq.~(\ref{eq: nd minmax empty state}) matches Eq.~(\ref{eqn: multiway min-max recurrence improper}), 
	so it suffices to show that Eq.~(\ref{eq: nd minmax nonempty state}) is correct.
	
	That Eq.~(\ref{eq: nd minmax nonempty state}) holds follows from a similar argument as in Section~\ref{sec: polynomial time algorithm multiway minmax}.
	Recall that we have a correspondence between $\treeplus_v$-cuts $\multicutQ_v$ (whose top component contains $v$, which is necessary if $\multicutQ_v$ is proper) 
	and pairs of a $\treeplus_u$-cut $\multicutQ_u$ and $\treeplus_w$-cut $\multicutQ_w$.
	Specifically, given $\multicutQ_u$ and $\multicutQ_w$ which are $L$-feasible, we can merge them to obtain $\multicutQ_v$ 
	which is $L$-feasible and satisfies $\multicutQ_v (\pi_v) = \braced{\pi_v} \cup \multicutQ_u(\pi_u) \cup \multicutQ_w(\pi_w)$.
	(We do not know a priori that $\multicutQ_v (\pi_v)$ doesn't contain a conflict, but we can still interpret $\multicutQ_v$ as a partition.)
	This implies $\cutof \multicutQ_v (\pi_v) = \cutof \multicutQ_u (\pi_u) \sqcup \cutof \multicutQ_w (\pi_w)$.
	Likewise, given an $L$-feasible $\multicutQ_v$ (whose top component contains $v$, which is necessary if $\multicutQ_v$ is proper), 
	we can obtain an $L$-feasible pair $\multicutQ_u$ and $\multicutQ_w$ satisfying the same relations.
	
	Importantly, in both directions the three cuts satisfy the relation $\Terminals_v = \Terminals_u \sqcup \Terminals_w$, 
	where $\Terminals_x = \multicutQ_x(\pi_x) \cap \Terminals$ denotes the terminals in the respective top components.
	This implies that if we started with $\multicutQ_v$ and $\typeof(\Terminals_v)=s\ne\emptyset$ ($\multicutQ_v$ is in state $s$), 
	then $(\typeof(\Terminals_u), \typeof(\Terminals_w))\in \splits(s)$ is a split of $s$ by Property~(A2).
	Conversely, if we started with $\multicutQ_u$ and $\multicutQ_w$ and $(\typeof(\Terminals_u), \typeof(\Terminals_w))\in \splits(s)$, 
	then, by Property~(A3), $\Terminals_v$ is guaranteed to be independent and satisfy $\typeof(\Terminals_v)=s$,
	thus $\multicutQ_v$ as the result of merging is in fact a $\treeplus_v$-cut in state $s$. 
	This correspondence then implies the correctness of Eq.~(\ref{eq: nd minmax nonempty state}), as desired.
	
	
	\paragraph*{Extension to Real Weights.}
	As in Section~\ref{sec: polynomial time algorithm multiway minmax}, for each $v\in V$ we can define the function $H_v$ and real number $L_v$,
	where $H_v(L)$ equals $\weight(v\pi_v)$ if $\weight(v\pi_v) + D_v^* (L) \le L$ and is $\infty$ otherwise,
	and $L_v\coloneqq \min\braced{L \suchthat \weight(v\pi_v)+D_v^*(L) \le L}$.
	Since $D_v^*(L)$ has the same semantics as the table entry $D_v^1(L)$ from Section~\ref{sec: polynomial time algorithm multiway minmax}, $H_v$ and $L_v$ have the same interpretation as before.
	In particular, $H_v$ is a right-continuous step function with at most one point of discontinuity at $L_v$.
	And by induction, $D_v^s$ (for $s\ne \emptyset$) and $D_v^*$, as functions of $L$, are non-increasing right-continuous functions whose discontinuities are a subset of $\jumps_v^1=\braced{L_u}_{u \in T_v-\braced{v}}$,
	while $D_v^\emptyset$ has the same form with discontinuities in $\jumps_v^0 = \jumps_v^1 \cup \braced{L_v}$.
	
	This implies that we can compute a list representation for each dynamic programming table entry, using $O(n)$ space per function and $O(n)$ time for computing component-wise sums and minimums.
	In the same way as in Section~\ref{sec: polynomial time algorithm multiway minmax}, the discontinuities $L_v$ are computed dynamically from the list representation of $D_v^*$.
	This completes the proof of Theorem~\ref{thm: generalized neighborhood diversity result}(b) for real weights.


\bigskip

\myparagraph{Information on the use of generative AI.}
The ideas of using the local cuts and the linear program~(\ref{eq: lp constraint}) in Section~\ref{sec: 2-approximation minmax}, 
and of working with 
the entire functions $D_v^\sigma$ in Section~\ref{sec: polynomial time algorithm multiway minmax}, 
were suggested by the LLM-based
service Bolzano. We thank Vašek Rozhoň for providing us with access to the service.
In addition, GPT-6 suggested the half-integrality proof in Appendix~\ref{sec: primal dual algorithm}.


\newpage
\bibliographystyle{abbrv}
\bibliography{connected_multicuts_references}


\vfill
\newpage
\appendix


\newpage
\section{Faster 2-Approximation for Unweighted Trees}
\label{sec: faster unweighted approximation}


In Section~\ref{sec: 2-approximation minmax} we gave a polynomial-time $2$-approximation algorithm for the
unweighted version of $\ConMulticutMaxTrees$. That algorithm was rather inefficient. In this section we
show an alternative and much faster algorithm.


\begin{theorem}\label{thm: faster 2-approximation for unweighted min-max}
There is a $2$-approximation algorithm for unweighted \normalfont{$\ConMulticutMaxTrees$}
with running time $O(|\ConflictEdges| + n\log n)$.
\end{theorem}

From now on, we fix an instance of $\ConMulticutMaxTrees$ consisting of
a normal-form tree $T=(V,E,\weight)$ (where $\weight(e)=1$ for all $e\in E$)
with root $r$ and a conflict graph $\ConflictGraph =(\Terminals, \ConflictEdges)$.
To prove Theorem~\ref{thm: faster 2-approximation for unweighted min-max}, it is sufficient to prove the
lemma below, as the algorithm from this lemma can be combined with binary search for $k$ to compute the $2$-approximate solution.


\begin{lemma}\label{lem: faster min-max 2-approximation key lemma}
Let $(T, \ConflictGraph)$ be an unweighted \normalfont{$\ConMulticutMaxTrees$} instance. 
After $O(n+|\ConflictEdges|)$ pre-processing, there is an $O(n)$ algorithm that, given $k\in \naturals$, either returns
a multicut of cost at most $2k$ or determines that $\maxopt(T,\ConflictGraph)>k$.
\end{lemma}

The rest of this section consists of the proof of Lemma~\ref{lem: faster min-max 2-approximation key lemma}.
For $v\in V$, we define a \emph{$T_v$-cut} to be a connected multicut $\multicutQ$ for the restricted instance $(T_v, \ConflictGraph \cap T_v)$, 
with the possible exception that the component containing $v$ (called the ``top component'' of $\multicutQ$)
may not contain a terminal.
$\multicutQ$ is called \emph{proper} if this component has a terminal and \emph{improper} otherwise. 
$\multicutQ$ is called \emph{$k$-feasible} if all non-top components have cost at most $k$.
Furthermore, the \emph{local cost} of $\multicutQ$ is the top component boundary weight (within $T_v$).
Let $\lambdaminus_v$ and $ \lambdaplus_v$ denote the optimum local cost over all $k$-feasible improper and proper $T_v$-cuts respectively.
We follow the convention that if $X$ is a component of a $T_v$-cut, $\cutof X$ always refers to the boundary edges in $T_v$, excluding the edge $v\pi_v$.
(This is similar to $\treeplus_v$-cuts from Section~\ref{sec: polynomial time algorithm multiway minmax}, 
with $\treeplus_v$ replaced by $T_v$.)

The essence of the algorithm is to extend Observation~\ref{obs: reduction from unweighted local cut to vertex cover} to a lower bound on $\lambdaminus_v$ and $\lambdaplus_v$.
Using this, we reduce the problem to solving a set of vertex cover instances, one for each sub-tree $T_v$, then perform a dynamic-programming
bottom-up computation on $T$.
For each vertex $v$, we will define an auxiliary graph $H_v$ whose nodes are some children of $v$, and
we use a $2$-approximate vertex cover for $H_v$ to compute the top component of a multicut for $T_v$.
Crucially, these vertex cover instances do not depend on $k$, and can thus be pre-processed, and their combined size does not exceed $n+|\ConflictEdges|$.


\subsection{Preprocessing}

Here we describe how to construct the vertex cover instances needed for our algorithm and how to solve them
all in $O(n+|\ConflictEdges|)$ time total.

For an internal vertex $v\in V$, let $A_v\subseteq ch(v)$ be the set of children
$u\in ch(v)$ of~$v$ such that the terminals in $T_u\cup\braced{v}$ do not conflict.
An \emph{auxiliary conflict graph} $H_v$ is $H_v = (A_v,F_v)$
where, for vertices $u,u' \in A_v$, $uu'\in F_v$ if and only if there exist terminals $t\in V(T_u)$
and $t'\in V(T_{u'})$ such that $t t'\in \ConflictEdges$. 

For vertices $u,v \in V$, let $\lca(u,v)$ denote their least common ancestor in $T$ (which may be $u$ or $v$).
For a vertex $v$ and descendant $u\in T_v-\braced{v}$, let $\child_u(v)$ be the unique child $w\in \children(v)$ such that $u \in T_w$.
Trivially, $\sum_{v \in V} |A_v| \le n$.
Furthermore, we claim that $\sum_{v\in V} |F_v| \le |\ConflictEdges|$.
This inequality follows by mapping each $tt' \in \ConflictEdges$ with $\lca(t,t')\notin\braced{t,t'}$
to $uu' \in F_v$, where $v=\lca(t,t')$, $u = \child_t(v)$ and $u' = \child_{t'}(v)$.

The pseudo-code below shows how to compute every $H_v$ in $O(n+|\ConflictEdges|)$ time.
Correctness follows directly from the definition of $H_v$ and the correspondence described above.
Since $\lca(u,v)$ and $\child_u(v)$ can be computed in constant time with $O(n)$ pre-processing using standard data structures, 
the computation of all graphs $H_v$ will take time $O(n+|\ConflictEdges|)$.


\begin{center}
\framebox{
\begin{minipage}{3in}
\begin{tabbing}
aa \= aa \= aa \= aa \= aa \= aa \= \kill
\textbf{Algorithm:} Computing graphs $H_v$ and vertex covers $S_v$\\
\> \myforeach $v\in V$ \mydo $a_v \gets \textfalse$, $A_v \gets \emptyset$, $F_v \gets \emptyset$ \\
\> \myforeach $tt' \in \ConflictEdges$ \mydo\\
\>\> 	$v \gets \lca(t, t')$ and $a_v \gets \texttrue$ \\
\>\> 	\myif $t = v$ \mythen $a_{u'}\gets \texttrue$ for $u' = \child_{t'}(v)$ \\
\>\> 	\myif $t' = v$ \mythen $a_{u}\gets \texttrue$ for $u = \child_t(v)$ \\
\> \myforeach $v\in V$ (processing bottom-up) \mydo\\ 
\>\> 	$a_v \gets a_v \lor \bigvee_{u \in \children(v)} a_u$ \\
\>\> 	\myif {$v\neq r$} and $a_v = \textfalse$ \mythen add $v$ to $A_{\pi_v}$ \\
\> \myforeach $tt' \in \ConflictEdges$ \mydo\\ 
\>\> 	$v \gets \lca(t, t')$ \\
\>\> 	\myif $v\notin \braced{t,t'}$ \mythen \\
\>\>\> 			$u\gets \child_t(v)$ and $u' \gets \child_{t'}(v)$ \\
\>\>\> 			\myif $u,u'\in A_v$ \mythen  add $uu'$ to $F_v$ \\
\> \myforeach internal $v\in V$ \mydo\\
\>\> $H_v \gets (A_v,F_v)$\\
\>\>\myif $F_v \neq \emptyset$ \mythen\\
\>\>\>\> $S_v \gets$ vertex cover of $H_v$ with $|S_v|\le \min\braced{ 2\cdot\vcover(H_v), |A_v|} - 1$ \\
\>\>\myelse $S_v = \emptyset$
\end{tabbing}
\end{minipage}
}
\end{center}

To compute $S_v$ in the last step, we sligtly modify the standard matching-based
$2$-approximation for $\VertexCover$:
First choose any edge $(t',t'')\in F_v$. (We know that $|F_v| \neq \emptyset$.)
Let $H'$ and $H''$ be the graphs obtained from $H_v$ by removing $t'$ and $t''$, respectively.
Compute $2$-approximate solutions $S'_v$ and $S''_v$ for $\VertexCover$ on $H'_v$ and $H''_v$.
If $|S'_v|\le |S''_v|$ then let $S_v = S'_v\cup \braced{t'}$, else let  $S_v = S''_v\cup \braced{t''}$.
Since one of $t',t''$ must be in the optimal vertex cover of $H_v$, we obtain that 
$|S_v|\le 2\cdot\vcover(H_v)-1$.
If it so happens that $S_v = A_v$ (and $A_v\neq \emptyset$), 
we can simply remove any vertex from $S_v$ so that we also have $|S_v|\le |A_v|-1$.
The running time for computing each $S_v$ is $O(|A_v| + |F_v|)$, so
the total cost of computing all sets $S_v$ amortizes to $O(n+|\ConflictEdges|)$.


\subsection{Main Algorithm}

Let $k\in\naturals$ be the given integer. If $k=0$, the algorithm checks if $\ConflictEdges =\emptyset$,
and if that is the case, it returns the optimal multicut consisting of just one component $\braced{T}$,
and otherwise it reports that  $\maxopt(T,\ConflictGraph)>0$.

So we assume now that $k\ge 1$. Our algorithm proceeds bottom-up, computing table entries $\Dplus_v$ and $\Dminus_v$,
for all $v\in V$, that contain a proper $T_v$-cut and an improper $T_v$-cut, respectively.
We will also keep track of the respective top components that we denote by $\Xplus_v$ and $\Xminus_v$.
Similar to Sections~\ref{sec: polynomial time algorithm multiway minmax}~and~\ref{sec: neighborhood diversity algorithms},
$\cutof \Xplus_v$, $\cutof \Xminus_v$, $\cutof \Dplus_v$, and $\cutof \Dminus_v$ are by convention restricted to $T_v$.
So $|\cutof \Xplus_v|$ and  $|\cutof \Xminus_v|$ are the local costs of $\Dplus_v$ and $\Dminus_v$.
(It's possible that the solution for $v$ will be undefined, which we represent by $\nullsol$. By default,
the corresponding local cost is then $\infty$.)

Computing solutions for the leaves is trivial:
if $t$ is a leaf (and thus $t\in\Terminals$), let $\Dplus_t=\braced{\Xplus_t}$ for $\Xplus_t = \braced{t}$, and $\Dminus_t = \Xminus_t = \nullsol$.
So assume now that $v$ is an internal vertex such that all vertices in $T_v\setminus\braced{v}$
have already been processed. Let $B_v=ch(v)\setminus A_v$. We partition $B_v$ into two subsets $\Bplus_v$ and $\Bminus_v$:
\begin{align*}
\Bplus_v \;&=\; \braced{ u \in B_v \suchthat |\cutof \Xplus_u| \le 2k-1}
\\
\Bminus_v \;&=\; B_v\setminus \Bplus_v 
\end{align*}
(Note that this partitioning depends on the algorithm.)

\myitparagraph{Computing $\Dplus_v$ and  $\Xplus_v$:} We distinguish three cases. 
\begin{description}
\item{(P1)} If $v\in \Terminals$ or $A_v \ne \emptyset$: Set
\begin{align*}
	\Xplus_v \;&=\; \braced{v} 
					\cup \bigcup_{u \in A_v\setminus S_v} V(T_u)
					\cup \bigcup_{u\in \Bminus_v} \Xminus_u
	\\
	\Dplus_v \;&=\; \braced{\Xplus_v} 
					\cup \braced{ V(T_u) \suchthat u\in S_v}
					\cup \bigcup_{u \in \Bplus_v} \Dplus_u 
					\cup \bigcup_{u\in \Bminus_v} (\Dminus_u \setminus\braced{\Xminus_u})
\end{align*}
\item{(P2)} Otherwise, if $\Bplus_v\ne\emptyset$: Select $\staru\in \Bplus_v$ for which $|\cutof \Xplus_{\staru}|$ is minimized, and set
\begin{align*}
	\Xplus_v \;&=\; \braced{v} 
					\cup \Xplus_{\staru}
					\cup \bigcup_{u\in \Bminus_v} \Xminus_u 
	\\
	\Dplus_v \;&=\; \braced{\Xplus_v} 
					\cup (\Dplus_{\staru}\setminus\braced{\Xplus_{\staru}})
					\cup \bigcup_{u \in \Bplus_v \setminus \braced{{\staru}}} \Dplus_u
					\cup \bigcup_{u\in \Bminus_v} (\Dminus_u \setminus\braced{\Xminus_u})
\end{align*}
\item{(P3)} Otherwise (that is, for $\Bminus_v = ch(v)$ and $v\notin\Terminals$): set $\Xplus_v,\Dplus_v=\nullsol$.
\end{description}

\myitparagraph{Computing $\Dminus_v$ and  $\Xminus_v$:} We have two cases.
\begin{description}
\item{(I1)} If $v \notin \Terminals$, set
\begin{align*}
	\Xminus_v \;&=\; \braced{v} 
					 \cup \bigcup_{u\in \Bminus_v} \Xminus_u 
	\\
	\Dminus_v \;&=\; \braced{\Xminus_v}
					\cup \braced{ V(T_u) \suchthat u\in A_v}
					\cup \bigcup_{u \in \Bplus_v} \Dplus_u
					\cup \bigcup_{u \in \Bminus_v} (\Dminus_u \setminus\braced{\Xminus_u})
\end{align*}
\item{(I2)} Otherwise (for $v \in \Terminals$), set $\Xminus_v,\Dminus_v = \nullsol$
\end{description}

In the event that $\nullsol$ appears in the formulas above, we set $Y \cup \nullsol = \nullsol$ for any set $Y$.
With this caveat, we are done describing our recurrence equations.
We observe that  (when not $\nullsol$) the $T_v$-cuts $\Dplus_v$ obtained via these
recurrences are proper. In Case~(P1),
the top component of $\Dplus_v$ either contains $v$ as a terminal or a terminal leaf from $A_v$.
In Case~(P2), it contains a terminal from the top component of $\Dplus_{\staru}$.
As for $\Dminus_v$, in Case~(I1) it does not contain a terminal, because $v$ is not a terminal and
top components $\Xminus_u$ for $u\in \Bminus_v$ also do not contain terminals.

After processing all vertices, the algorithm returns $\Dplus_r$ if $|\cutof \Xplus_r| \le 2k-1$ 
and otherwise declares $\maxopt(T, \ConflictGraph)>k$. This completes the description of the algorithm. 

\smallskip

Regarding the running time,
for a linear-time implementation, we only need to keep track of the edges in $\Xplus_v,\Xminus_v$ which are incident to $v$.
Then by standard backtracking we can obtain the final multicut.


\myparagraph{Analysis.}
To complete the proof of Lemma~\ref{lem: faster min-max 2-approximation key lemma}, it remains to prove its correctness.
We first prove a conditional bound on local costs for each $T_v$-cut in Lemma~\ref{lemma: 2-approx lambda invariant},
then in Lemma~\ref{lemma: 2-approx boundary size} we show that bounding the local costs of our computed cuts yields a bound on the min-max cost.

Throughout the rest of the proof, for a real number $a$, we use notation $\posvalue{a} = \max(a,0)$.


\begin{lemma}\label{lemma: 2-approx lambda invariant}
For each vertex $v$:
\begin{description}
	\item{(a)} If $\lambdaplus_v \le k$, then $|\cutof \Xplus_v| \le \posvalue{2\lambdaplus_v-1}$. 
	\item{(b)} Likewise, $\lambdaminus_v \le k$ implies $|\cutof \Xminus_v| \le \posvalue{2 \lambdaminus_v-1}$. 
\end{description}
\end{lemma}

By the normal form of $T$, each leaf is a terminal, which implies that $\lambdaminus_v\ge 1$ for each node $v$
($\lambdaminus_v$ could be $\infty$), so in~(b) we can write $2 \lambdaminus_v-1$ instead of $\posvalue{2 \lambdaminus_v-1}$
on the right-hand side of the equation.

\begin{proof}
We proceed by induction on the tree, assuming \emph{both} implications in our induction hypothesis.
The lemma is trivial when $v$ is a leaf, so assume $v$ is an internal vertex and the lemma holds for all $u \in ch(v)$ (for both solution types).
	
First we make one general observation about the children of $v$ in $\Bminus_v$.

\begin{claim}\label{claim: 2-approx Type-C children}
If $\multicutQ'$ is a $k$-feasible $T_v$-cut (either proper or improper) with local cost at most $k$,
then for each $u\in \Bminus_v$, $v$'s component in $\multicutQ'$ contains $u$ but it does not contain any terminal from $T_u$. 
\end{claim}

\begin{claimproof}
Let $X'\in\multicutQ'$ be the component containing $v$.
If $X'$ did not contain $u$ then, by $k$-feasibility of $\multicutQ'$, $\multicutQ'$ restricted to $T_u$ would be a proper $k$-feasible $T_u$-cut with local cost at most $k$. 
(In particular, the top component in this $T_u$-cut would be a non-top component in $\multicutQ'$, so it would contain a terminal.)
Similarly, if $X'$ contained $u$ and some terminal from $T_u$
then again, $\multicutQ'$ restricted to $T_u$ would be a proper $k$-feasible $T_u$-cut with local cost at most $k$. 
(In this case, the top component is $X'\cap T_u$ and contains a terminal.)
Either way, we have $\lambdaplus_u \le k$, and by our induction hypothesis, 
this implies $|\cutof \Xplus_u| \le \posvalue{2\lambdaplus_u -1} \le 2k-1$, contradicting the definition of $\Bminus_v$.
\end{claimproof}
	
As a corollary, our assumption that $\min\braced{\lambdaminus_v, \lambdaplus_v} \le k$ implies
$\lambdaminus_u \le \min\braced{\lambdaminus_v, \lambdaplus_v} \le k$ 
(and $|\cutof \Xminus_u| \le  2\lambdaminus_u-1$ by induction) for all $u \in \Bminus_v$.


\myitparagraph{Proof of~(a): computing $\Dplus_v$.}
We first analyze $\Dplus_v$, breaking into cases corresponding to the cases in its computation. 
For the next few claims, it is convenient to introduce the following notation.
Given a component $X\in\multicutQ$ of a $T_v$-cut and $u\in ch(v)$, let $\cutof_u X = \cutof X \cap (E(T_u) \cup \braced{uv})$.
Then let $\cutof_A X = \bigsqcup_{u \in A_v} \cutof_u X$, 
$\cutof_\Bplus X = \bigsqcup_{u \in \Bplus_v} \cutof_u X$, and $\cutof_\Bminus X = \bigsqcup_{u \in \Bminus_v} \cutof_u X$.
Clearly we have $\cutof X = \cutof_A X \sqcup \cutof_\Bplus X \sqcup \cutof_\Bminus X$.


\begin{claim}\label{claim: 2-approx vertex cover bound}
If $\lambdaplus_v \le k$ then
\begin{equation*}
		\lambdaplus_v \;\ge\; \vcover(H_v) + |\Bplus_v| + \sum_{u \in \Bminus_v} \lambdaminus_u
\end{equation*}
\end{claim}

\begin{claimproof}
Consider a proper $k$-feasible $T_v$-cut $\multicutQ'$ with local cost $\lambdaplus_v \le k$, and let $X'\in\multicutQ'$ be the top component.
We produce three bounds on $\cutof_A X'$, $\cutof_\Bplus X'$, and $\cutof_\Bminus X'$.
The set $S'=\braced{u \in A_v \mid |\cutof_u X'|\ge 1 }$ is a vertex cover for $H_v$, implying $|\cutof_A X'| \ge \vcover(H_v)$.
(If edge $uw \in F_v$ was not covered by $S'$, then $\cutof_u X'=\cutof_w X'=\emptyset$ and $X$ contains both $T_u$ and $T_w$.
But $uw\in F_v$ implies there is some conflicting terminal pair $tt'\in\ConflictEdges$ where $t\in T_u$ and $t'\in T_w$, implying $X'$ has a conflict.)
If $u\in \Bplus_v$, then $|\cutof_u X'| \ge 1$, otherwise $X'$ would contain $T_u \cup \braced{v}$ and a conflicting pair.
Therefore $|\cutof_\Bplus X'| \ge |\Bplus_v|$.
If $u\in \Bminus_v$, then by Claim~\ref{claim: 2-approx Type-C children}, $\multicutQ'$ induces an improper $k$-feasible $T_u$-cut with local cost $|\cutof_u X'|$,
implying $|\cutof_u X'| \ge \lambdaminus_u$ and $|\cutof_\Bminus X'| \ge \sum_{u \in \Bminus_v} \lambdaminus_u$.
Summing inequalities yields our desired bound.
\end{claimproof}

\medskip
Suppose now that $v$ satisfies Case (P1) of the algorithm, that is $v\in \Terminals$ or $A_v \neq \emptyset$. In this case,
\begin{equation*}
		|\cutof \Xplus_v| \;=\; |S_v| + |\Bplus_v| + \sum_{u \in \Bminus_v} |\cutof \Xminus_u| 
\end{equation*}
If $F_v\neq\emptyset$ then $|S_v| \le 2\razy\vcover(H_v)-1$.
If $\Bplus_v\neq\emptyset$ then $|\Bplus_v| \le 2|\Bplus_v| -1$.
If $\Bminus_v\neq\emptyset$ then for all $u\in \Bminus_v$ we have $|\cutof \Xminus_u| \le 2\lambdaminus_u -1$, by induction.
Therefore, if at least one of these conditions holds,  Claim~\ref{claim: 2-approx vertex cover bound} implies that 
\begin{equation*}
		|\cutof \Xplus_v| \;\le\; 2\razy \vcover(H_v) + 2|\Bplus_v| + 2\sum_{u \in \Bminus_v} \lambdaminus_u -1
		\;\le\; 2\lambdaplus_v -1
\end{equation*}
On the other hand, if all three conditions are false then $\cutof \Xplus_v = \emptyset$. 
Thus $\cutof \Xplus_v \le \posvalue{2\lambdaplus_v -1}$  in Case~(P1), as needed.


The next claim handles the case when $v\notin \Terminals$ and $A_v=\emptyset$.

\begin{claim}\label{claim: 2-approx Type-B bound}
If $\lambdaplus_v \le k$, $v\notin\Terminals$ and $A_v =\emptyset$, then $\Bplus_v \ne \emptyset$ and
\begin{equation*}
\lambdaplus_v \ge |\Bplus_v|-1 + \min_{u \in \Bplus_v} \lambdaplus_u + \sum_{u \in \Bminus_v} \lambdaminus_u
\end{equation*}
\end{claim}

\begin{claimproof}
Consider a proper $k$-feasible $T_v$-cut $\multicutQ'$ with local cost $\lambdaplus_v\le k$, and let $X' \in \multicutQ'$ be the top component.
By similar reasoning as before, Claim~\ref{claim: 2-approx Type-C children} implies $|\cutof_u X'| \ge \lambdaminus_u$ for $u\in \Bminus_v$, so $|\cutof_\Bminus X'| \ge \sum_{u \in \Bminus_v} \lambdaminus_u$.
Combined with $v\notin \Terminals$ and $A_v=\emptyset$, Claim~\ref{claim: 2-approx Type-C children} also implies that $X'$ contains a terminal from $T_w$ for some $w \in \Bplus_v$.
(In particular, $\Bplus_v \ne \emptyset$.)
Then $\multicutQ'$ induces a proper $k$-feasible $T_w$-cut of local cost $|\cutof_w X'|$, implying $|\cutof_w X'| \ge \lambdaplus_w \ge \min_{u \in \Bplus_v} \lambdaplus_u$.
We have $|\cutof_u X'| \ge 1$ for the remaining children $u\in \Bplus_v-\braced{w}$, which implies our desired bound.
\end{claimproof}

\smallskip
Suppose now that $v$ satisfies Case~(P2) of the algorithm, that is $v\notin \Terminals$, $A_v = \emptyset$, and $\Bplus_v\ne\emptyset$.
Claim~\ref{claim: 2-approx Type-B bound} implies there is some $w\in \Bplus_v$ 
where $\lambdaplus_w = \min_{u \in \Bplus_v} \lambdaplus_u \le \lambdaplus_v \le k$.
By the algorithm, we have  $|\cutof \Xplus_{\staru}| \le |\cutof \Xplus_{w}|$, so
\begin{equation*}
		|\cutof \Xplus_v| \;\le\; |\cutof \Xplus_w| + |\Bplus_v|-1 + \sum_{u \in \Bminus_v} |\cutof \Xminus_u| 
\end{equation*}
Now, if $\lambdaplus_w \ge 1$ then, by induction, $|\cutof \Xplus_w|\le 2\razy\lambdaplus_w-1$.
If $|\Bplus_v|\ge 2$ then $|\Bplus_v|-1 \le 2( |\Bplus_v|-1 ) - 1$.
If $\Bminus_v\neq\emptyset$ then for all $u\in \Bminus_v$ we have $|\cutof \Xminus_u|\le 2 \lambdaminus_u-1$, by induction.
So if at least one of these three conditions holds then
\begin{equation*}
		|\cutof \Xplus_v| \;\le\; 2\lambdaplus_w + 2(|\Bplus_v|-1) + 2\sum_{u \in \Bminus_v} \lambdaminus_u -1 \;\le\; 2\lambdaplus_v -1
\end{equation*}
On the other hand, if all three conditions are false then $\cutof \Xplus_v = \emptyset$,
So $|\cutof \Xplus_v| \le \posvalue{2\lambdaplus_v -1}$ holds in Case~(P2).

\smallskip
Regarding Case~(P3), note that Claim~\ref{claim: 2-approx Type-B bound} implies that $A_v \cup \Bplus_v$ is non-empty, 
which means that (under the assumption that $\lambdaplus_v\le k$)
the algorithm will never execute Case~(P3) for computing the proper $T_v$-cut where $\Bminus_v=ch(v)$.
This concludes the analysis of $\Dplus_v$.


\myitparagraph{Proof of~(b): computing $\Dminus_v$.}
We now analyze $\Dminus_v$ and derive one more bound.

\begin{claim}\label{claim: 2-approx Type-0 bound}
If $\lambdaminus_v\le k$ then
\begin{equation*}
	\lambdaminus_v \ge |A_v| + |\Bplus_v| + \sum_{u \in \Bminus_v} \lambdaminus_u
\end{equation*}
\end{claim}

\begin{claimproof}
Let $\multicutQ'$ be an improper $k$-feasible $T_v$-cut with local cost $\lambdaminus_v \le k$, and let $X' \in \multicutQ'$ be the top component.
By similar reasoning as before, Claim~\ref{claim: 2-approx Type-C children} implies $|\cutof_\Bminus X'| \ge \sum_{u \in \Bminus_v} \lambdaminus_u$.
For $u \in A_v \cup \Bplus_v$, $|\cutof_u X'| \ge 1$, for otherwise $X'$ would contain $T_u$ and its terminals.
(Recall that $T$ is in normal form, so all leaves are terminals.)
\end{claimproof}

For computing $\Dminus_v$ the algorithm has two cases. In Case~(I2), $v\in\Terminals$
implies that $\lambdaminus_v = \infty$, and the inductive step holds.
So consider Case~(I1), that is $v\notin\Terminals$. Then
\begin{equation*}
	|\cutof \Xminus_v| \;=\; |A_v| + |\Bplus_v| + \sum_{u \in \Bminus_v} |\cutof \Xminus_u| 
\end{equation*}
If $A_v\cup \Bplus_v\neq\emptyset$ then $|A_v| + |\Bplus_v|\le 2(|A_v| + |\Bplus_v|) - 1$.
If $\Bminus_v\neq\emptyset$ then, as in the previous cases, for $u\in \Bminus_v$ we have $\lambdaminus_u\ge 1$, so
the induction implies $|\cutof \Xminus_u| \le 2\lambdaminus_u-1$.
One of these cases must hold (since $v$ is not a leaf), so
\begin{equation*}
	|\cutof \Xminus_v| \;\le\; 2|A_v| + 2|\Bplus_v| + 2\sum_{u \in \Bminus_v} \lambdaminus_u -1
			\;\le\; 2\lambdaminus_v -1
\end{equation*}
This finishes the proof of Lemma~\ref{lemma: 2-approx lambda invariant}.
\end{proof}


\begin{lemma}\label{lemma: 2-approx boundary size}
For each vertex $v$:
\begin{description}
\item{(a)} If $|\cutof \Xplus_v| \le 2k-1$ then  $\Dplus_v$ has max-cost (in $T_v$) at most $2k$. 
\item{(b)} Likewise,
if $|\cutof \Xminus_v| \le 2k-1$ then  $\Dminus_v$ has max-cost (in $T_v$) at most $2k$. 
\end{description}
\end{lemma}

\begin{proof}
We proceed by induction on $T$, again assuming both implications in the induction hypothesis.
The lemma is trivial for leaves, so suppose $v$ is an internal vertex. 
Whether $|\cutof \Xplus_v|$ or $|\cutof \Xminus_v|$ is at most $2k-1$, we have the following bounds on the children's local costs:
$|\cutof \Xplus_u| \le 2k-1$ for $u\in \Bplus_v$ by definition and $\sum_{u \in \Bminus_v} |\cutof \Xminus_u| \le \min\braced{|\cutof \Xplus_v|, |\cutof \Xminus_v|}$ by construction.
This means our induction hypothesis applies to all children $u\in B_v$. Specifically, for $u \in \Bplus_v$ and $u' \in \Bminus_v$,
$\Dplus_u$ and $\Dminus_{u'}$ have each max-cost at most $2k$ in their respective sub-trees.

We now prove part~(a), by estimating the max-cost of $\Dplus_v$ when $|\cutof \Xplus_v| \le 2k-1$. (The analysis for part~(b) is nearly identical, so we omit it.) 
Let $X$ be a component of $\Dplus_v$, so $\cutof X$ denotes $X$'s boundary in $T_v$.
If $X$ is the top component (i.e.\ $v\in X$), then $|\cutof X|\le 2k-1$, by our assumption.

So we can now assume that $X$ is a non-top compoment of $\Dplus_v$.
If $\cutof X$ contains an edge $uv$ where $u\in A_v$, then $X=T_u$ and $|\cutof{X}|=1$.
If instead $uv\in \cutof X$ and $u\in \Bplus_v$, then $X$ is the top component of $\Dplus_u$ and $|\cutof X|=|\cutof \Xplus_u|+1 \le 2k$.
Lastly, if $X$ does not satisfy any of the previous conditions, then $X$ must be a non-top component of a previously computed $T_u$-cut $\multicutQ$, for some $u\in B_v$.
Either $u\in \Bplus_v$ and $\multicutQ=\Dplus_u$, or $u\in \Bminus_v$ and $\multicutQ=\Dminus_u$; either way, by induction we have $|\cutof{X}|\le 2k$.
\end{proof}

If $\maxopt(T,\ConflictGraph) \le k$, then $\lambdaplus_r \le \maxopt(T,\ConflictGraph) \le k$, 
which by Lemma~\ref{lemma: 2-approx lambda invariant} implies $|\cutof \Xplus_r| \le 2k-1$. 
Then by Lemma~\ref{lemma: 2-approx boundary size}, $\Dplus_r$ is a multicut of $(T, \ConflictGraph)$ with max-cost at most $2k$, as desired.


\newpage
\section{Normal Forms}
\label{sec: normal forms proof}

In Section~\ref{sec: preliminaries} we defined a normal form and a binary normal form
for trees and stated Lemma~\ref{lemma: normal form for trees} that says that
each tree can be converted into an appropriate normal form.
In this section we prove this lemma.

We first describe how to modify a non-weighted tree $T$ into normal form,
which can be done by repeatedly applying two reduction rules:
(1) If $v$ is a non-terminal leaf, remove it;
(2) If $v$ is a non-terminal with exactly two neighbors $u$ and $w$ (and $v$ is not the root), add an edge between $u$ and $w$ and delete $v$. 
These steps will eventually convert $T$ into a tree $T'$ in normal form.
To justify correctness, consider just one step, converting $T$ into a tree $T''$, and show that it satisfies all necessary 
properties except possibly~(a). For instance, in rule~(2), a multicut in $T$ that contains $uv$ or $wv$ (it cannot contain both)
has the same cost as the multicut $T''$ that contains $uw$ instead.

The above construction also works for the weighted case: in rule~(2) simply let $\weight''(uw)=\min\braced{\weight(uv), \weight(vw)}$.
Thus to show part (b), we can assume that $T$ is in the normal form already.
Let $W=\sum_{e \in E} \weight(e)$. To obtain $T'$ we then repeatedly apply the following rules:
\begin{description}
\item{(3)} If $t\in \Gamma$ is an internal vertex with children $u_1,...,u_d$, temporarily remove $t$ and create a new vertex $v$ in its place. 
That is, $v$ is now the parent of each $u_i$ with the same weight $\weight'(vu_i)=\weight(tu_i)$. Then reattach $t$ to $v$ as a leaf and set $\weight'(tv)=W$. 
\item{(4)} If $v\notin \Gamma$ is an internal vertex with children $u_1,...,u_d$ for $d\ge 3$, create a new vertex $w$ and make it the new parent of $u_2,...,u_d$ 
with the same weights $\weight'(wu_i)=\weight(vu_i)$. Then make $w$ a child of $v$ and set $\weight'(wv)=W$. 
\end{description}
These steps will eventually transform $T$ into binary normal form. To justify correctness, 
let $T''$ be the tree obtained after one step of this process from $T$, and argue that it has the required properties, except possibly~(b).
The idea, in both cases~(3) and~(4), is that 
$W$ is so large that any multicut of $T''$ that cuts $tv$ can be modified to cut other edges without increasing cost, so
again we have an obvious correspondence between multicuts of $T$ and $T''$.  
Both rules preserve the normal form conditions, so the resulting tree will eventually be in binary normal form.


\newpage
\section{$\NP$-Hardness of Svitkina/Tardos' $\TreeCutting$ Problem~\cite{svitkina_tardos_min-max_multiway_cuts_2004}}
\label{sec: np-hardness tree cutting}


Here we consider the $\TreeCutting$ problem of Svitkina and Tardos~\cite{svitkina_tardos_min-max_multiway_cuts_2004}
which is essentially a bi-objective combination of $\ConMultiwayCutSumTrees$ and $\ConMultiwayCutMaxTrees$:
Given a weighted tree $T = (V,E,\weight)$, a set of terminals $\Terminals\subseteq V$, and a bound $B$, the objective
is to compute a partition of $T$ into subtrees $T_1,...,T_h$, each with at most one terminal and with $\weight(\boundary T_i) \le B$,
that minimizes total boundary weight $\half\sum_{i=1}^h \weight( \boundary T_i)$.
Svitkina and Tardos~\cite{svitkina_tardos_min-max_multiway_cuts_2004} gave a pseudo-polynomial algorithm for $\TreeCutting$
that they used as a subroutine to design an approximation algorithm for the $\MultiwayCutMaxTrees$ problem for weighted trees
(without the connectivity requirement). They did not analyze, however, the computational complexity of $\TreeCutting$
itself. Here, we show that the problem is $\NP$-hard, and thus a pseudo-polynomial time algorithm is all we can hope for
if we insist on finding the exact optimum.


\begin{figure}[t]
\begin{center}
\includegraphics[width = 3in]{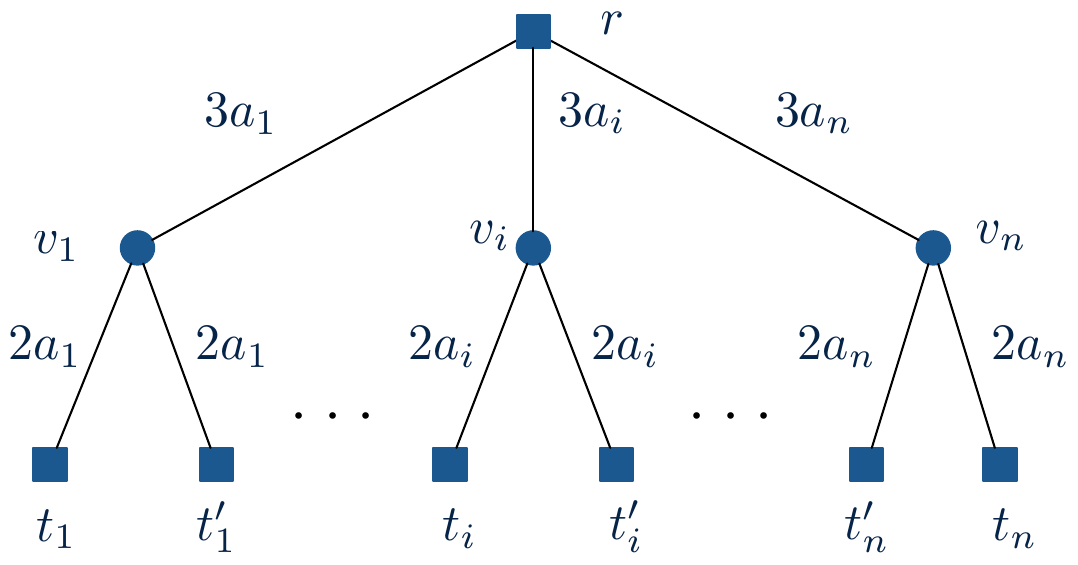}
\caption{%
The constructed instance of $\TreeCutting$. Squares represent terminals, circles represent non-terminals.}
\label{fig: tree cutting np-hardness}
\end{center}
\end{figure}


The proof is by reduction from $\Partition$. Consider an instance of $\Partition$ with integers $a_1,...,a_n$ and let $A = \half\sum_{i=1}^n a_i$.
We convert it into the following instance of $\TreeCutting$: the tree has a root vertex $r$, and for each $i = 1,...,n$ it has three
vertices $v_i$, $t_i$ and $t'_i$, and edges $(r,v_i)$ with weight $3a_i$, and $(v_i,t_i)$, $(v_i,t'_i)$, both with weight $2a_i$. 
(See Fig.~\ref{fig: tree cutting np-hardness}.)
The terminals are $\Terminals = \braced{r,t_1,t'_1,...,t_n,t'_n}$ and the boundary bound is $B = 7A$.
We claim that
\begin{description}
\item{$(\ast)$} The instance of $\Partition$ has a solution iff the constructed instance of $\TreeCutting$ has a solution of
	cost at most $9A$.
\end{description}
The proof of $(\ast)$ is quite simple, so we only sketch the idea, leaving the details to the reader. In the instance of
$\TreeCutting$, for each $i$, in order to separate terminals $r$, $t_i$ and $t'_i$
each partition must cut at least two edges among $(r,v_i)$, $(v_i,t_i)$, and $(v_i,t'_i)$. Cutting all three unnecessarily
increases the total cost, so we can assume that only two are cut. By symmetry, we can then assume that $(v_i,t_i)$ is
cut. So we have two types of cuts: Type~1 with $(v_i,t_i)$ and $(r,v_i)$  cut, and Type~2 with $(v_i,t_i)$ and $(v_i,t'_i)$ cut, at costs
$5a_i$ and $4a_i$, respectively. 
Have the cut of Type~1 represent $i$ being in the solution $I$ of the instance of $\Partition$, and Type~2 of being in the
complement $\barI$.
All clusters except the top one (containing $r$) have small cost, so only the top cluster is of
concern. Then to satisfy the condition on $\TreeCutting$ in $(\ast)$ we need:
\begin{equation*}
	\textstyle
	\sum_{i\in I} 3a_i + \sum_{i\in \barI} 4a_i \;\le \; 7A
	\quad\textrm{and}\quad
	\sum_{i\in I} 5a_i + \sum_{i\in \barI} 4a_i \;\le \; 9A
\end{equation*}
By simple algebra, these inequalities are equivalent to $\sum_{i\in I} a_i = A$, so the equivalence $(\ast)$ is true.


\newpage
\section{On the Connectivity Requirement for $\MultiwayCutMax$ }
\label{sec: connectivity requirement for min-max}


In this section we examine how the connectivity requirement for $\MulticutMax$ affects the value of the optimum
solution. We focus on the simple case in which the conflict graph is a clique, that is on $\MultiwayCutMax$.
(Recall that in $\MultiwayCutMax$ the conflict graph is a clique, so all terminals must be separated.)
We refer to the worst-case ratio between the optimal solution for $\MultiwayCutMax$ and the optimal connected
solution (namely the solution for $\ConMultiwayCutMax$) as the \emph{connectivity gap}.


\myparagraph{Trees.}
We start with the case of weighted trees, that is the solution of $\MultiwayCutMaxTrees$ versus that of $\ConMultiwayCutMaxTrees$.
Svitkina and Tardos~\cite{svitkina_tardos_min-max_multiway_cuts_2004} gave an example of an instance showing that for trees the
connectivity gap for $\MultiwayCutMaxTrees$ is at least $1.25$. 
We show here that this gap is $\Theta(\log n)$.

For the lower bound, we use a non-weighted and perfectly balanced binary tree $T$
and the set $\Terminals$ of terminals that are the leaves of $T$. 
Consider a partition of $T$'s vertices into pairs, where each
pair is formed by grouping each non-terminal with a unique terminal (one leaf will not be paired with any internal vertex). 
The min-max cost of this partition is $O(1)$.
On the other hand,
in a connected partition, the part containing the root would also contain the path $P$ from the root to the unique terminal $t$ of this part.
Then any cut that separates $P$ from $\Terminals\setminus\braced{t}$ must have $\Omega(\log n)$ edges.
(We can assume that $P$ be the leftmost path from the root. Then for each
vertex $v\in P\setminus\braced{t}$, the rightmost path from $v$ to a leaf must contain at least one edge in this cut.)
Thus the boundary weight of the partition containing the root is $\Omega(\log n)$, and we can conclude
that the connectivity gap, even for non-weighted trees, is $\Omega(\log n)$.

We next argue that connectivity gap for weighted trees is $O(\log n)$. 
Consider an instance consisting of a tree $T = (V,E,\weight)$ and a set of terminals $\Terminals$.
Let $\multicutM$ be an optimal multicut of $\MultiwayCutMaxTrees$ for this instance, and
let $L = \maxcost(\multicutM)$. 
(Here $\multicutM$ is represented as a partition into connected components, each without conflicts. 
For brevity, here we will refer to these connected components simply as \emph{components},
even though some may not contain a terminal.)
We will convert $\multicutM$ into a connected multicut $\multicutM'$ satisfying
$\maxcost(\multicutM') = O(L \log|\Terminals|)$, thus establishing the claimed upper bound on the connectivity gap.

Each component of $\multicutM$ contains at most one terminal.
We refer to the pieces with terminals as \emph{terminal components} and other as \emph{nonterminal components}.
To obtain $\multicutM'$, we combine the components into larger subtrees, each
containing exactly one terminal.
To this end, consider a graph $Q$ whose nodes are the components of $\multicutM$ and edges connect the pairs of components that have
an edge (clearly, exactly one) of $T$ between them.  Then $Q$ is itself a tree, which we root at an arbitrary node $r$.

First, repeatedly merge any nonterminal leaf into its parent.
This operation does not increase the boundary weight of the parent component.
Consequently, every remaining component has boundary weight at most $L$, and every leaf of $Q$ is a terminal component.
Next, ensure that all nonterminal components of degree two are non-adjacent. To do this, 
consider any inclusion-wise maximal path of degree two nonterminal components and merge them together into one subtree.
The boundary weight of this new subtree is at most $2L$.
After repeating this process, each nonterminal component in $Q$ of degree two is only adjacent to terminal components.

For each nonterminal component $C$, choose a nearest terminal node among its descendants,
and retain only the edge from $C$ to the child on the path to that terminal
(i.e., delete every other edge to $C$ and from $C$). Remove from $Q$ also all edges
connecting two terminal components. This way, by the properties of $Q$ described in the previous paragraph,
the tree $Q$ is partitioned
into paths of length at most $\log_2 |\Gamma|$, each ending at a terminal component.
Now merge the components of each such path, producing a subtree of $T$ that contains exactly one terminal.
The boundary of this subtree will have weight $O(L \log |\Terminals|)$, completing the proof.


\myparagraph{Arbitrary graphs.}
We now show that for arbitrary graphs, the connectivity gap for $\MultiwayCutMax$ could be much larger than for trees, namely at least $\Omega(\sqrt{n})$,
even in the non-weighted case.
Our example is a non-weighted graph $G = (V,E)$ with terminals $\Terminals\subseteq V$.
Let $V' = V\setminus \Terminals$ be the set of non-terminals.
$\Terminals$ is partitioned into two sets, $\Terminals = \Terminals' \cup \Terminals''$,
where $|\Terminals'| = k'$ and $|V'| = |\Terminals''| = \ell$, so that $n = k'+2\ell$.
The topology of $G$ is this: the vertices in $V'\cup \Terminals'$ form a clique of size $k'+\ell$, and the vertices in $\Terminals''$ are isolated.

Let $\multicutM$ be a partition representing the following multicut: each terminal in $\Terminals'$ is a singleton, and all other vertices
are grouped into $\ell$ disjoint pairs, with one vertex from $V'$ and one from $\Terminals''$. Then $\maxcost(\multicutM) = \ell+k'-1 \le n$.

To show our lower bound, it is sufficient to prove that for some choice of $k'$ and $\ell$ each connected multicut has cost $\Omega(n^{1.5})$.
In any connected multicut, the components of the terminals in $\Terminals''$ are singletons, and
$V'\cup \Terminals'$ is partitioned into components, each containing one terminal from $\Terminals'$.
Let $B$ be a component with maximum cardinality.
Observe that $|\cutof B| = |B|(\ell+k'-|B|)$.
Choosing $k' = \sqrt{n}$ and $\ell = \half ( n - \sqrt{n})$, with $n$ sufficiently large,
we have $|B| = \Omega(\sqrt{n})$, which implies that  the boundary weight of $B$ is $|\cutof B| = \Omega(n^{1.5})$.
(To see this, consider the cases when $|B|\le n/4$ and $|B| > n/4$.)
Thus the cost of this multicut is $\Omega(n^{1.5})$, completing the proof of our $\Omega(\sqrt{n})$ lower bound.


\newpage
\section{Primal-Dual Algorithm for $\LocalCutTreesImprop$ }
\label{sec: primal dual algorithm}

In this section we present a primal-dual algorithm for $\LocalCutTreesImprop$ from Section~\ref{sec: 2-approximation minmax},
providing an alternative proof of Lemma~\ref{lem: two approximation for local cut}.
Let $(T, \ConflictGraph, \starv)$ be our instance.
The dual of the linear program in the proof of Lemma~\ref{lem: two approximation for local cut} is as follows:
\begin{align}
	 \text{maximize } \sum_{st \in \ConflictEdges} \alpha_{st} & \;
	&\\\
	\text{subject to }\sum_{st \in \ConflictEdges} \alpha_{st} (1_{e \in P_{s\starv}} + 1_{e \in P_{t\starv}}) 
				\;&\le\; \weight(e)  \quad \forall e \in E 
	\label{eq: dual constraint}
	\\
	 \alpha_{st} \;&\ge\; 0 \quad \forall st \in \ConflictEdges
	\nonumber
\end{align}
Naturally, we can interpret the dual as a flow problem: for each conflict edge $st\in\ConflictEdges$,
we send \emph{two} flows from $\starv$, one to $s$ along $P_{s\starv}$ and another to $t$ along $P_{t \starv}$.
We send $\alpha_{st}$ units of flow along both paths.
(The intersection $P_{t\starv}\cap P_{s\starv}$ thus receives $2\alpha_{st}$ units of flow.)
The right hand side of inequality~(\ref{eq: dual constraint}) is then the total flow going through edge $e$.

We greedily construct a dual solution, similar to the primal-dual algorithms for $\VertexCover$ and $\MulticutSumTrees$.
The variables $\alpha_{st}$ are initially zero, and we say a tree edge $e \in E$ is \emph{saturated} when its corresponding constraint is tight.
Then for each conflict edge $st \in \ConflictEdges$,
we simply raise the value of $\alpha_{st}$ until a tree edge (possibly many) in $P_{s\starv} \cup P_{t\starv}$ become saturated.
Let $J\subseteq E$ be the set of saturated tree edges after applying this process to all conflict edges.
Our local cut $X$ is then the connected component in $T-J$ containing $\starv$.

By rooting $T$ at $\starv$, we can obtain the tree edges $P_{s\starv} \cup P_{t\starv}$ for any $st\in\ConflictEdges$ in $O(n)$ time,
and so this algorithm can be implemented in $O(n|\ConflictEdges|)$ time.
Furthermore, $X$ is indeed a local cut:
for each conflict edge $st \in \ConflictEdges$, $\cutof X$ must contain at least one saturated edge from $P_{s\starv} \cup P_{t\starv}$ by construction,
implying that at least one of $s,t$ is not in $X$.

We now show that $\weight(\cutof X)\le 2\sum_{st \in \ConflictEdges} \alpha_{st}$, where $\alpha_{st}$ refers to the constructed dual variables.
The boundary $\cutof X \subseteq J$ contains only (but not necessarily all) saturated edges,
and since $X$ is connected and contains $\starv$, we have $|\cutof X \cap P_{t \starv}| \le 1$ for all $t\in \Terminals$.
This implies
\begin{align*}
	\weight(\cutof X) &= \sum_{e \in \cutof X} \weight(e) =  \sum_{e \in \cutof X} \sum_{st \in \ConflictEdges} \alpha_{st} (1_{e \in P_{s\starv}} + 1_{e \in P_{t\starv}}) \\
	&=   \sum_{st \in \ConflictEdges} \alpha_{st} \sum_{e \in \cutof X}(1_{e \in P_{s\starv}} + 1_{e \in P_{t\starv}}) \le 2 \sum_{st\in\ConflictEdges} \alpha_{st}
\end{align*}
as desired. Then, by duality, $X$ is a 2-approximate solution.

\begin{figure}[t]
\begin{center}
\includegraphics[width = 2.5in]{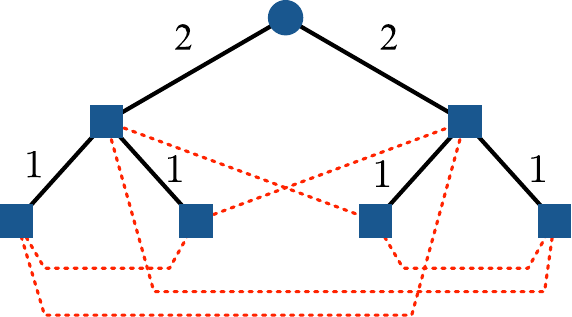}
\caption{An instance of the ``non-split LP'' with non-half-integral extremal points; $\starv$ is the single non-terminal at the top, and conflicts are illustrated in dashed red lines.
One can use an LP solver to verify that the unique
optimal solution of cost $8/3$ is achieved by setting $x_e =1/3$ for all edges $e$.
By inspection, the best half-integral and integral solutions have cost $3$.
}
\label{fig: nont-split LP counter example}
\end{center}
\end{figure}

\myparagraph{Half-Integrality.} We remark that the linear program from Section~\ref{subsec: local cut} is also half-integral.
In fact, the distances $x(P_{t\starv})$ for each $t\in \Terminals$ lie in $\braced{0,1/2,1}$ at the polytope's extremal points.
To see this, we root $T$ at $\starv$ and observe that our LP is equivalent to the following program:
\begin{align}
	 \text{minimize } & \;\sum_{v \in V-\braced{\starv}} \weight(v\pi_v)(x_v - x_{\pi_v}  )
	&\\\
	\text{subject to } x_v \;&\ge\; x_{\pi_v}  \quad \forall v \in V - \braced{\starv} \label{eq: order constraint}
	\\
	x_{s} + x_{t} \;&\ge\; 1 \quad \forall st \in \ConflictEdges \\
	 x_{\starv}=0, x_v \;&\ge\; 0 \quad \forall v \in V - \braced{\starv}
	\nonumber
\end{align}
where our variables are $\braced{x_v}_{v \in V-\braced{\starv}}$ (with $x_{\starv}=0$ as a dummy variable).
Intuitively, we've simply done a change of variables, where now we assign distances to vertices that induces edge weights $x(v\pi_v) = x_v - x_{\pi_v}$,
using the observation that weights of a tree can be reconstructed from distances from a fixed source.
To ensure that the induced weights are non-negative, we impose an ordering constraint in Eq.~\ref{eq: order constraint}.
Clearly, extremal points of our original LP correspond to extremal points of this new one.

Ignoring the slightly more complicated objective function,
it's now clear that our new LP constraints are just a modified version of the $\VertexCover$ LP.
Indeed, standard arguments proving the half-integrality of the $\VertexCover$ LP (e.g.~\cite{nemhauser_trotter_properties_of_vertex_packing_and_independence_system_polyhedra_1974}) also apply here;
the extra ordering constraints make no difference.

Interestingly, we can slightly modify the linear program from Section~\ref{subsec: local cut} by not double-counting edges in $P_{s\starv} \cap P_{t \starv}$,
which simply means replacing Eq.~\ref{eq: lp constraint} with $\sum_{e \in P_{s \starv} \cup P_{t \starv} } x_e \ge 1$.
This ``non-split LP'' is still a relaxation of $\LocalCutTreesImprop$, and the same LP rounding and primal-dual algorithms apply to this variant.
However, this variant is \emph{not} half-integral. An example is shown in Figure~\ref{fig: nont-split LP counter example}, and indeed it's not to hard to construct for every integer $d \ge 3$ an example
containing an extremal point with denominator $d$.

\end{document}